\documentclass[runningheads]{llncs}

\usepackage{latexsym}
\usepackage{amsfonts}
\usepackage{amssymb}
\usepackage{amsmath}
\usepackage{color}
\usepackage{textcomp}
\usepackage{tikz}
\usepackage{lineno,hyperref} 
\usepackage{graphicx}
\usepackage{stmaryrd} 
\usepackage{fancybox}
\usepackage{tikz}
\usetikzlibrary{patterns,decorations,decorations.pathmorphing,%
arrows,shapes,automata,backgrounds,fit,calc,petri}
\usepackage[T1]{fontenc}
\usepackage{mathrsfs}
\usepackage{mathptmx}
\usepackage{helvet}

\newtheorem{notation}{Notation}

\DeclareMathOperator{\phase}      {phase}

\DeclareMathOperator{\proj}       {proj}
\DeclareMathOperator{\mixsseq}    {mixsseq}
\DeclareMathOperator{\maxmixsseq} {maxmixsseq}
\DeclareMathOperator{\sseq}       {sseq}
\DeclareMathOperator{\maxsseq}    {maxsseq}
\DeclareMathOperator{\fseq}       {fseq}
\DeclareMathOperator{\scenarios}  {scenarios}
\DeclareMathOperator{\maxscenarios}{maxscenarios}
\DeclareMathOperator{\scenario}   {scenario}
\DeclareMathOperator{\finreach}   {finreachable}
\DeclareMathOperator{\reach}      {reachable}
\DeclareMathOperator{\pree}       {pre}
\DeclareMathOperator{\postt}      {post}
\DeclareMathOperator{\MARK}       {markings}
\DeclareMathOperator{\phcMARK}    {phcmarkings}
\DeclareMathOperator{\STEPS}      {steps}
\DeclareMathOperator{\enabled}    {enabled}
\DeclareMathOperator{\Csan}       {csan}
\DeclareMathOperator{\syncycles}  {syncycles}

\newcommand{\drop}[1] { }

\newcommand{\pset}[1] {\mathbb{P}(#1)}
\newcommand{\post}[1] {#1^\bullet}
\newcommand{\pre}[1]  {{}^\bullet{#1}}
\newcommand{\POST}[2] {\postt_{#1}{(#2)}}
\newcommand{\PRE}[2]  {\pree_{#1}{(#2)}}
\newcommand{\STEP}[2] {\boldsymbol{[}#1\boldsymbol{\rangle}_{\!#2}\,}
\newcommand{\mixSTEP}[2]
     {\boldsymbol{[}#1\boldsymbol{\rangle}\!\boldsymbol{\rangle}_{\!#2}\,}
\newcommand{\an}      {\textit{acnet}}
\newcommand{\han}     {\textit{hanet}}
\newcommand{\lan}     {\textit{lanet}}
\newcommand{\bdan}    {\textit{bdanet}}
\newcommand{\on}      {\textit{ocnet}}
\newcommand{\net}     {\textit{net}}
\newcommand{\bdcsan}  {\textit{bdcsan}}
\newcommand{\hcsan}   {\textit{hcsan}}
\newcommand{\lcsan}   {\textit{lcsan}}
\newcommand{\hcson}   {\textit{hcson}}
\newcommand{\lcson}   {\textit{lcson}}
\newcommand{\csan}    {\textit{csan}}
\newcommand{\cson}    {\textit{cson}}

\newcommand{\bsan}    {\textit{bsan}}
\newcommand{\bson}    {\textit{bson}}
\newcommand{\init}    {\textit{init}}
\newcommand{\fin}     {\textit{fin}}
\newcommand{\es}      {\varnothing}
\newcommand{\EOD}     {\hfill $\diamondsuit$}
\newcommand{\id}      {\textit{id}}
\newcommand{\AN}      {\textit{AN}}
\newcommand{\CSAN}    {\textit{CSAN}}
\newcommand{\CSON}    {\textit{CSON}}
\newcommand{\BDCSAN}  {\textit{BDCSAN}}
\newcommand{\WFCSAN}  {\textit{WFCSAN}}
\newcommand{\BSAN}    {\textit{BSAN}}
\newcommand{\BSON}    {\textit{BSON}}

\newcommand{\WFBSAN}  {\textit{WFBSAN}}
\newcommand{\ON}      {\textit{ON}}
\newcommand{\BDAN}    {\textit{BDAN}}
\newcommand{\WFAN}    {\textit{WFAN}}

\pgfdeclarelayer{background}
\pgfdeclarelayer{foreground}
\pgfsetlayers{background,main,foreground}

\newcommand{\DiagramBig}[1][0.5]
     {
     \begin{tikzpicture}[node distance=1.3cm,>=arrow30,%
     line  width=0.3mm,scale=#1,bend angle=45]
     \tikzstyle{box}=[draw,regular polygon,thick,%
     regular polygon sides=4,minimum size=22mm, inner sep = -3pt]
     \tikzstyle{ybox}=[draw,regular polygon,thick,%
     regular polygon sides=4,minimum size=22mm, inner sep = -3pt,fill=yellow]
     \tikzstyle{rbox}=[draw,regular polygon,thick,%
     regular polygon sides=4,minimum size=22mm, inner sep = -3pt,fill=red]
     \tikzstyle{gbox}=[draw,regular polygon,thick,%
     regular polygon sides=4,minimum size=22mm, inner sep = -3pt,fill=green]

     \tikzstyle{dybox}=[draw,regular polygon,thick,style=dashed,%
     regular polygon sides=4,minimum size=22mm, inner sep = -3pt,fill=yellow]
     \tikzstyle{drbox}=[draw,regular polygon,thick,style=dashed,%
     regular polygon sides=4,minimum size=22mm, inner sep = -3pt,fill=red]
     \tikzstyle{dgbox}=[draw,regular polygon,thick,style=dashed,%
     regular polygon sides=4,minimum size=22mm, inner sep = -3pt,fill=green]
     }

\newcommand{\Diagram}[1][0.5]
     {
     \begin{tikzpicture}[node distance=1.3cm,>=arrow30,%
     line  width=0.3mm,scale=#1,bend angle=45]
     \tikzstyle{box}=[draw,regular polygon,thick,rounded corners,%
     regular polygon sides=4,minimum size=14mm, inner sep = -3pt]
     \tikzstyle{ybox}=[draw,regular polygon,thick,rounded corners,%
     regular polygon sides=4,minimum size=14mm, inner sep = -3pt,fill=yellow!40]
     \tikzstyle{rbox}=[draw,regular polygon,thick,rounded corners,%
     regular polygon sides=4,minimum size=14mm, inner sep = -3pt,fill=red!40]
     \tikzstyle{gbox}=[draw,regular polygon,thick,rounded corners,%
     regular polygon sides=4,minimum size=14mm, inner sep = -3pt,fill=green!40]
     \tikzstyle{bbox}=[draw,regular polygon,thick,rounded corners,%
     regular polygon sides=4,minimum size=14mm, inner sep = -3pt,fill=cyan!40]

     \tikzstyle{dybox}=[draw,regular polygon,thick,style=dashed,rounded corners,%
     regular polygon sides=4,minimum size=14mm, inner sep = -3pt,fill=yellow!40]
     \tikzstyle{drbox}=[draw,regular polygon,thick,style=dashed,rounded corners,%
     regular polygon sides=4,minimum size=14mm, inner sep = -3pt,fill=red!40]
     \tikzstyle{dgbox}=[draw,regular polygon,thick,style=dashed,rounded corners,%
     regular polygon sides=4,minimum size=14mm, inner sep = -3pt,fill=green!40]
     \tikzstyle{dbbox}=[draw,regular polygon,thick,style=dashed,rounded corners,%
     regular polygon sides=4,minimum size=14mm, inner sep = -3pt,fill=cyan!40]
     }

\newcommand{\Relation}[2]
     {
     \begin{tikzpicture}[baseline=-0.5cm, node distance=.8cm,>=arrow30,line width=0.3mm]
     \tikzstyle{tlcorner}=[xshift=#2,yshift=#1]
     \tikzstyle{brcorner}=[xshift=#1,yshift=#2]
     \tikzstyle{bgcolor}=[fill=brown!20]
     }

\newcommand{\StandardNet}[1][0.5]
     {
     \begin{tikzpicture}[node distance=1.3cm,>=latex',line  width=0.3mm,scale=#1,auto,bend angle=45]
     \tikzstyle{place}=[draw,circle,thick,minimum size=4.5mm]
     \tikzstyle{transition}=
                [draw,regular polygon,thick,
                 regular polygon sides=4,minimum size=8mm, inner sep = -2pt]
     }

\newcommand{\StandardTS}[1][0.5]
     {
     \begin{tikzpicture}[node distance=0.5cm,>=stealth',bend angle=45,scale=#1]
     }

\newcommand{\dirDOTS}     [2]{\path (#1) edge [dotted](#2);}

\newcommand{\diredge}     [2]{\path (#1) edge [->] (#2);}

\newcommand{\PlacN}[5]{\node(#1)at(#2,#3)[place,tokens=0,label=above:$#5$]      {$#4$};}
\newcommand{\PlacS}[5]{\node(#1)at(#2,#3)[place,tokens=0,label=below:$#5$]      {$#4$};}

\newcommand{\bPlacE}[5]{\node(#1)at(#2,#3)[place,tokens=0,label=right:$#5$,ultra thick]      {$#4$};}
\newcommand{\bPlacW}[5]{\node(#1)at(#2,#3)[place,tokens=0,label=left :$#5$,ultra thick]      {$#4$};}

\newcommand{\Whitetran}[4]{\node (#1) at (#2,#3)[transition] {$#4$};}

\newcommand{\Put}[3]{
	\node at (#1,#2) {#3};
}

\pgfarrowsdeclare{arrow30}{arrow30}
{
  \arrowsize=0.2pt
  \advance\arrowsize by.275\pgflinewidth%
  \pgfarrowsleftextend{+-\arrowsize}
  \advance\arrowsize by.5\pgflinewidth
  \pgfarrowsrightextend{+\arrowsize}
}
{
  \arrowsize=0.2pt
  \advance\arrowsize by.275\pgflinewidth%
  \pgfsetdash{}{+0pt}
  \pgfsetroundjoin
  \pgfpathmoveto{\pgfqpoint{1\arrowsize}{0\arrowsize}}
  \pgfpathlineto{\pgfqpoint{-12\arrowsize}{2.5\arrowsize}}
  \pgfpathlineto{\pgfqpoint{-12\arrowsize}{-2.5\arrowsize}}
  \pgfpathclose
  \pgfusepathqfillstroke
}

	\titlerunning {Structured Acyclic Nets}
	\authorrunning{Structured Acyclic Nets}

\title{Structured Acyclic Nets} 
\author{Mohammed Alahmadi,
                  Salma Alharbi,
                  Talal Alharbi,
                  Nadiyah Almutairi,
                  Tuwailaa Alshammari,
                  Anirban Bhattacharyya,
                  Maciej Koutny, 
                  Bowen Li, and 
                  Brian Randell}
\institute{School of Computing, Newcastle University\\Urban Sciences Building, 
1 Science Square, 
Newcastle Helix\\Newcastle upon Tyne, NE4 5TG, United Kingdom
\\\today}
 
\begin{document}

\maketitle

\begin{abstract}
The concept of structured occurrence nets 
is an extension of that of occurrence nets which are directed acyclic 
graphs that represent causality and
concurrency information concerning a single execution of a distributed
system. The formalism of structured occurrence nets has been introduced to facilitate
the portrayal and analysis of the behaviours, and in particular
failures, of complex evolving systems. 
Such systems are composed of a large number of sub-systems which may proceed
concurrently and interact with each other and with the external
environment while their behaviour is subject to modification
by other systems.
 
The  purpose of this paper is to provide an extension of structured occurrence nets 
to include models built up of acyclic nets rather than occurrence nets.  
\end{abstract}

\section{Introduction}

The concept of structured occurrence nets~\cite{c1,c51} 
is an extension of that of occurrence nets~\cite{DBLP:journals/tcs/BestD87} which are directed acyclic 
graphs (a subclass of Petri nets) that represent causality and
concurrency information concerning a single execution of a distributed
system. The formalism of structured occurrence nets has been introduced to facilitate
the portrayal and analysis of the behaviours, and in particular
failures, of complex evolving systems. 
 Examples include a large hardware system
which suffers component break-downs, reconfigurations and
replacements, a large distributed system whose software is
being continually updated (or patched), a gang of criminals
whose membership is changing, and an operational railway
system that is being extended. (In these latter cases we are
regarding crimes and accidents as types of failure.) The underlying idea of a structured occurrence net is to combine multiple related occurrence nets by
using various formal relationships (in particular, abstractions)
in order to express dependencies among the component occurrence nets.
By means of these relations, a structured occurrence net is able to portray a more
explicit view of system activity, involving various types of
communication between subsystems, and of system upgrades,
reconfigurations and replacements than is possible with an
occurrence net, so allowing one to document and exploit behavioural
knowledge of (actual or envisaged) complex evolving systems.

Communication structured occurrence nets are
a basic variant of structured occurrence nets that enable
the explicit representation of synchronous and asynchronous
interaction between communicating subsystems.  A communication structured occurrence net 
is composed of a set of distinct component
occurrence nets representing separate subsystems. When it is determined
that there is a potential for an interaction between subsystems,
an asynchronous or synchronous communication link can be
made between events in the different subsystems’ occurrence nets via 
a channel/buffer place.
 
Another variant of structured occurrence nets, behavioural
structured occurrence nets, conveys information
about the evolution of individual systems. They use a two
level view to represent an execution history, with the lower
level providing details of its behaviours during the different
evolution stages represented in the upper level view. Thus a
behavioural
structured occurrence net gives information about the evolution of an individual
system, and the phases of the overall activity are used to
represent each successive stage of the evolution of this system. 
 
This document extends and systematises the ideas contained in~\cite{c1}, after
allowing backward non-de\-ter\-mi\-nism and forward non-de\-ter\-mi\-nism in the representation of 
the components of a complex evolving system. This is achieved by replacing 
occurrence nets with more general acyclic nets.
As a result, it leads to communication structured acyclic nets and 
behavioural
structured acyclic nets generalising previously introduced models.
 
\section{Preliminaries}

All sets used in the relational structures 
considered in this paper are finite which simplifies some of the
definitions and results.
For two sets $X$ and $Y$, $X\subset Y$ means that $X$ is a proper subset of
$Y$, i.e., $X\subseteq Y$ and $X\neq Y$.
The disjoint union of  $X$ and $Y$ is denoted by $X\uplus Y$, and 
nonempty
sets $X_1,\dots, X_k$ form a partition of a sets $X$ if 
$X=X_1\uplus \dots\uplus  X_k$.

For a binary relation $R$, $xRy$ means that $(x,y)\in R$.
The \emph{composition} of two binary relations, $R$ and $Q$,
is a binary relation given
by $R\circ Q\triangleq \{(x,y)\mid \exists z: xRz\wedge zQy\}$.
Moreover, for every $k\geq 1$, w define:
\[
    R^k
    \triangleq
    \left\{
    \begin{array}{l@{~~~}l}
    R & \textit{if}~k=1
    \\
    R\circ R^{k-1} & \textit{otherwise}\;.
    \end{array}
    \right.
\]

The first definition introduces notions related to partial orders
which can be used, e.g., to capture causality is concurrent behaviours.

\begin{definition}[relations and orderings]
\label{def:1}
    Let $X$ be a set and $R\subseteq X\times X$.
\begin{enumerate}
\item
    $\id_X\triangleq\{(x,x)\mid x\in X\}$ is the \emph{identity}
    relation on $X$.
\item
    $R$ is \emph{reflexive}
    if $\id_X\subseteq R$.
\item
    $R$ is \emph{irreflexive}
    if $R\cap\id_X=\es$.
\item
    $R$ is \emph{transitive}
    if $R\circ R\subseteq R$.
\item
    $(X,R)$ is a \emph{partial order} if $R$ is irreflexive
    and transitive.
\item
    $R^+\triangleq R^1\cup R^2 \cup\dots$
    is the \emph{transitive closure} of $R$.
\item
    $R$ is \emph{a\-cyc\-lic} if $R^+\cap \textit{id}_X=\es$.
\EOD
\end{enumerate}
\end{definition}

The following facts follow directly from the definitions.

\begin{proposition}
\label{prop:1}
    Let $X$ be a set and $R\subseteq X\times X$.
\begin{enumerate}
\item
    If $(X,R)$ is a partial order
    then $R$ is a\-cyc\-lic.
\item
    If $R$ is a\-cyc\-lic then
    $(X,R^+)$ is a partial order.
\item
    $R$ is a\-cyc\-lic iff
    there are no $x_1,\dots,x_k(= x_1)\in X$ ($k\geq 2$) such that
    $x_1 R x_2,\dots,x_{k-1} R x_k$.
\end{enumerate}
\end{proposition}

Rather than using sequences of events to represent possible
executions of concurrent systems, we will use
sequences of sets of events. The following definition introduces
two useful notions related to such sequences.

\begin{definition}[sequence of sets]
\label{def:1aa}
    Let $\sigma=X_1\dots X_k$ be a sequence of sets and $X$
    be a set.
\begin{enumerate}
\item
    $\bigcup\sigma\triangleq X_1\cup\dots\cup X_k$ denotes the
    set of elements \emph{occurring} in $\sigma$.
\item
    $\sigma|_X\triangleq (X_1\cap X)\dots (X_k\cap X)$ denotes the
    \emph{restriction} of $\sigma$ to the elements of $X$.
\item
    $\sigma\upharpoonright_X$ is obtained from $\sigma|_X$ by deleting
    all the occurrences of the empty set.
\EOD
\end{enumerate}
\end{definition}

\begin{example}
\label{ex-0s}
    For $\sigma=\{a,b\}\{a,c\}\{d\}$,
    $\bigcup\sigma=\{a,b,c,d\}$,
    $\sigma|_{\{a,b\}}=\{a,b\}\{a\}\es$, and
    $\sigma\upharpoonright_{\{a,b\}}=\{a,b\}\{a\}$.
\EOD
\end{example}

When defining or proving properties of the elements
of a set $X$ for which we are given an
a\-cyc\-lic relation $R$, one can  apply the principle of mathematical
induction. Typical patterns are as follows:

\begin{description}
\item
[defining property $P_x$ for every $x\in X$]{\ }
\\
First define $P_x$  for every $x$ such that
$\{y\mid yRx\}=\es$.
Then, for every other $x\in X$, define $P_x$ assuming
that $P_x$ has been defined for each element $y\in X$
such that $yR^+x$.
\item
[proving property $P_x$ for every $x\in X$]{\ }
\\
First prove $P_x$ for every $x$ such that
$\{y\mid yRx\}=\es$.
Then, for every other $x\in X$, prove $P_x$ assuming
that $P$ holds for each element $y\in X$
such that $yR^+x$.
\end{description}

\section{A\-cyc\-lic nets}

This section is concerned with Petri nets which
generalise occurrence nets and can provide  a direct support for
causality analysis.

The following are intuitive explanations of the main
concepts defined in this and subsequent sections
using terminology
based on accident/crime investigations:
\begin{description}

\item[a\-cyc\-lic net]{\ }
\\
A basic fragment  of `database' of all facts (both actual or hypothetical 
represented using places,
transitions, and arcs linking them)
accumulated during an investigation. Transitions (events) and
places (conditions/local states)
are related through  arrows representing causal and/or temporal
dependencies.
Hence, the
representation is required to be a\-cyc\-lic (this is a minimal requirement).
Acyc\-lic nets can represent alternative ways of interpreting
what has happened, and so may exhibit (backward and forward) non-de\-ter\-mi\-nism.

\item[backward de\-ter\-mi\-nistic a\-cyc\-lic net]{\ }
\\
An a\-cyc\-lic net such that for each event it is possible to state
precisely
which other events must have preceded it.

\item[occurrence net]{\ }
\\
An a\-cyc\-lic net providing a complete  record of all causal
dependencies between events  involved in a single `causal history'.
\end{description}

\begin{figure}[ht]
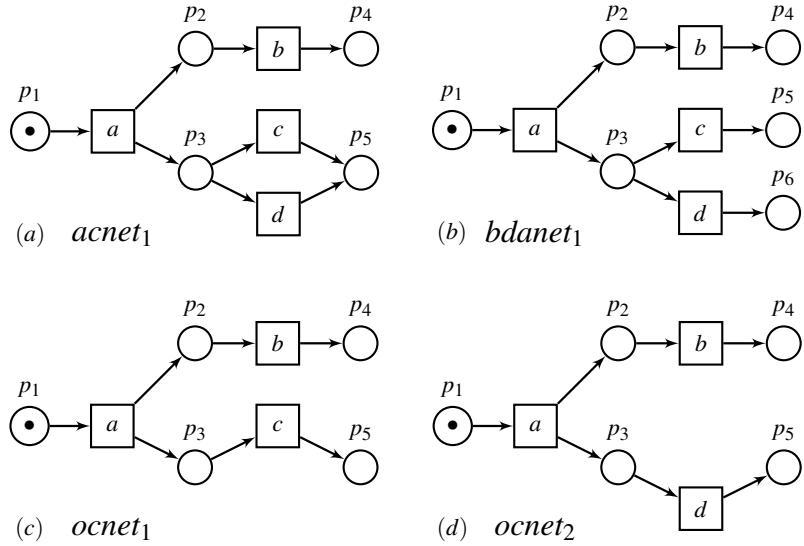

\begin{center}

\begin{tabular}{l@{~~~~~~~}l}
\StandardNet[0.55]
    \PlacN{P1}{ 0}{2}{\bullet}{p_1}
    \PlacN{P2}{ 4}{4}{}{p_2}
    \PlacN{P3}{ 4}{1}{}{p_3}
    \PlacN{P4}{ 8}{4}{}{p_4}
    \PlacN{P5}{ 8}{1}{}{p_5}

    \Whitetran{A}{ 2}{2}{a}
    \Whitetran{B}{ 6}{4}{b}
    \Whitetran{C}{ 6}{2}{c}
    \Whitetran{D}{ 6}{0}{d}

    \diredge{P1}{A}
    \diredge{P2}{B}
    \diredge{P3}{C}
    \diredge{P3}{D}

    \diredge{A}{P2}
    \diredge{A}{P3}
    \diredge{B}{P4}
    \diredge{C}{P5}
    \diredge{D}{P5}
\Put{0}{-.5}{$(a)$}
\Put{2}{-.5}{{\large$\an_1$}}
\end{tikzpicture}
&
\StandardNet[0.55]
    \PlacN{P1}{ 0}{2}{\bullet}{p_1}
    \PlacN{P2}{ 4}{4}{}{p_2}
    \PlacN{P3}{ 4}{1}{}{p_3}
    \PlacN{P4}{ 8}{4}{}{p_4}
    \PlacN{P5}{ 8}{2}{}{p_5}
    \PlacN{P6}{ 8}{0}{}{p_6}

    \Whitetran{A}{ 2}{2}{a}
    \Whitetran{B}{ 6}{4}{b}
    \Whitetran{C}{ 6}{2}{c}
    \Whitetran{D}{ 6}{0}{d}

    \diredge{P1}{A}
    \diredge{P2}{B}
    \diredge{P3}{C}
    \diredge{P3}{D}

    \diredge{A}{P2}
    \diredge{A}{P3}
    \diredge{B}{P4}
    \diredge{C}{P5}
    \diredge{D}{P6}
\Put{0}{-.5}{$(b)$}
\Put{2}{-.5}{{\large$\bdan_1$}}
\end{tikzpicture}
\\[4mm]
\StandardNet[0.55]
    \PlacN{P1}{ 0}{2}{\bullet}{p_1}
    \PlacN{P2}{ 4}{4}{}{p_2}
    \PlacN{P3}{ 4}{1}{}{p_3}
    \PlacN{P4}{ 8}{4}{}{p_4}
    \PlacN{P5}{ 8}{1}{}{p_5}

    \Whitetran{A}{ 2}{2}{a}
    \Whitetran{B}{ 6}{4}{b}
    \Whitetran{C}{ 6}{2}{c}

    \diredge{P1}{A}
    \diredge{P2}{B}
    \diredge{P3}{C}

    \diredge{A}{P2}
    \diredge{A}{P3}
    \diredge{B}{P4}
    \diredge{C}{P5}
\Put{0}{-.5}{$(c)$}
\Put{2}{-.5}{{\large$\on_1$}}
\end{tikzpicture}
&
\StandardNet[0.55]
    \PlacN{P1}{ 0}{2}{\bullet}{p_1}
    \PlacN{P2}{ 4}{4}{}{p_2}
    \PlacN{P3}{ 4}{1}{}{p_3}
    \PlacN{P4}{ 8}{4}{}{p_4}
    \PlacN{P5}{ 8}{1}{}{p_5}

    \Whitetran{A}{ 2}{2}{a}
    \Whitetran{B}{ 6}{4}{b}
    \Whitetran{D}{ 6}{0}{d}

    \diredge{P1}{A}
    \diredge{P2}{B}
    \diredge{P3}{D}

    \diredge{A}{P2}
    \diredge{A}{P3}
    \diredge{B}{P4}
    \diredge{D}{P5}
\Put{0}{-.5}{$(d)$}
\Put{2}{-.5}{{\large$\on_2$}}
\end{tikzpicture} 
\end{tabular}
\end{center}
\caption{A\-cyc\-lic nets with initial markings.}
\label{fig-1}
\end{figure}

\begin{definition}[a\-cyc\-lic net]
\label{def:2}
    An \emph{a\-cyc\-lic net} is a
    triple $\an=(P,T,F)$, where $P$ and $T$ are disjoint
    finite sets of \emph{places} and \emph{transitions}
    respectively,
    and $F\subseteq(P\times T)\cup(T\times P)$ is the
    \emph{flow relation} such that:
\begin{enumerate}
\item
    $P$ is nonempty and $F$ is a\-cyc\-lic.
\item
    For every $t\in T$, there are
    $p,q\in P$
    such that $pFt$ and $tFq$.
\end{enumerate}
    Notation: $\AN$ is the set of all a\-cyc\-lic nets.
\EOD
\end{definition}
Graphically, places are represented by circles,
transitions by boxes, and arcs between the nodes (i.e., places
and transitions)
represent the flow relation.
If it is important to
indicate explicitly  $\an$, we denote
$P$, $T$, $F$  by
$P_\an$, $T_\an$, $F_\an$, respectively.

In addition to the a\-cyc\-licity of $F$, it is required that each event
has at least one pre-condition (pre-place) and at least one post-condition (post-place).
A\-cyc\-lic net can exhibit backward non-de\-ter\-mi\-nism
(more than one arrow incoming to a place) as well as forward
non-de\-ter\-mi\-nism
(more than one arrow outgoing from a place).
 
\begin{notation}[direct precedence in a\-cyc\-lic net]
\label{not:1da}
Let $\an$ be an a\-cyc\-lic net.
To indicate relationships between different nodes, for all $x\in P_\an\cup T_\an$
and $X\subseteq P_\an\cup T_\an$, we denote the
\emph{directly preceding} and \emph{directly following} nodes as follows:
\[
\begin{array}{lllll@{~~~~}lllll}
  \pre{x}
    & =
    & \PRE{\an}{x}
    & \triangleq
    & \{z \mid zF_\an x\}
    & \pre{X}
    & =
    & \PRE{\an}{X}
    & \triangleq
    & \bigcup\{\pre{z}\mid z\in X\}
\\
  \post{x}
    & =
    & \POST{\an}{x}
    & \triangleq
    & \{z \mid xF_\an z\}
    & \post{X}
    & =
    & \POST{\an}{X}
    & \triangleq
    & \bigcup\{\post{z}\mid z\in X\}\;.
\end{array}
\]
Moreover, the \emph{initial} and \emph{final} places   are respectively
given by:
\[
    P_\an^\init\triangleq\{p\in P\mid \pre{p}=\es\}
    ~~\mbox{and}~~
    P_\an^\fin\triangleq\{p\in P\mid \post{p}=\es\} \;.
\]
\end{notation}
Note that having  the notations like $\pre{x}$ in addition to more
explicit $\PRE{\an}{x}$
helps to keep some of the subsequent formulas short. 
 
\begin{proposition}
\label{prop-vvcc}
    $P_\an=P_\an^\init\uplus\POST{\an}{T_\an}
    =P_\an^\fin\uplus\PRE{\an}{T_\an} $,
   for every  a\-cyc\-lic net $\an$.
\end{proposition}
\begin{proof}
    It follows directly from the definitions.
\qed
\end{proof}

\begin{example}
\label{ex-1eeeeeeeee}
In Figure~\ref{fig-1}($a$), $\an_1$ is an a\-cyc\-lic net such that
$\pre{p_5}=\PRE{\an_1}{p_5}=\{c,d\}$ and
$\post{a}=\POST{\an_1}{a}=\{p_2,p_3\}$.
Moreover, 
$P_{\an_1}^\init=\{p_1\}$ and $P_{\an_1}^\fin=\{p_4,p_5\}$.
\EOD
\end{example}

\subsection{Subnets of a\-cyc\-lic nets}

It is often desirable to analyse substructures of 
records represented by nets where only some of the events are
included.

\begin{definition}[subnet of a\-cyc\-lic net]
\label{def:4adddaa}
    A triple $\net=(P,T,F)$ is a \emph{subnet}
    of an a\-cyc\-lic net  $\an$ 
    if 
    $\es\neq P\subseteq P_\an$, 
    $T\subseteq T_\an$, 
    $F= F_\an|_{(P\times T)\cup (T\times P)}$, and,
    for every $t\in T$:
\[
    \{p\mid (p,t)\in F\}=\PRE{\an}{t} 
    ~~\mathit{and}~~
    \{p\mid (t,p)\in F\}=\POST{\an}{t}\;.
\] 
    Notation: $\net\subseteq\an$.
\EOD
\end{definition}
Note that a transition included in a subnet 
retains its pre-places and post-places; in other words, it retains 
its local environment. 

\begin{proposition}
\label{prop-vvc}
    A subnet of an a\-cyc\-lic net is also an a\-cyc\-lic net.
\end{proposition}
\begin{proof}
    It follows directly from the definitions and the fact that 
    a subset of an a\-cyc\-lic relation is also an a\-cyc\-lic relation.
\qed
\end{proof}

The next notion captures not only a structural inclusion
between a\-cyc\-lic nets, but it is
also intended to correspond to inclusion
of the behaviours they capture.

\begin{definition}[co-initial subnet of a\-cyc\-lic net]
\label{def:4aaa} 
    An a\-cyc\-lic net 
    $\an$ is a \emph{co-initial subnet} of an a\-cyc\-lic net $\an'$
    if $\an\subseteq\an'$ and 
    $P_\an^\init=P_{\an'}^\init$.
    \\
    Notation: $\an\sqsubseteq\an'$.
\EOD
\end{definition}

\begin{proposition}
\label{pr-uusu}
    Let $\an\sqsubseteq\an'$ be a\-cyc\-lic nets.
\begin{enumerate} 
\item
    $P_\an=P_{\an'}^\init\uplus \POST{\an'}{T_\an}$.
\item
    $\an=\an'$ iff $T_\an= T_{\an'}$.
\end{enumerate}
\end{proposition}
\begin{proof}
(1)
    By Proposition~\ref{prop-vvc}, $\an$ is an a\-cyc\-lic net, and so,
    by Proposition~\ref{prop-vvcc},  we have
    $P_\an=P_\an^\init\uplus \POST{\an}{T_\an}$.
    Hence, by Definitions~\ref{def:4adddaa}
    and~\ref{def:4aaa}, $P_\an=P_{\an'}^\init\uplus \POST{\an'}{T_\an}$. 
    
(2)
    The ($\Longrightarrow$) implication is obvious.
    To show the ($\Longleftarrow$) implication,
    suppose that $T_\an= T_{\an'}$. 
    Then, by part (1) and Proposition~\ref{prop-vvcc},
\[
    P_\an=P_{\an'}^\init\uplus \POST{\an'}{T_\an}=
    P_{\an'}^\init\uplus \POST{\an'}{T_{\an'}}=P_{\an'}\;.
\]
    Moreover, by Definition~\ref{def:4adddaa},
    $F_\an= F_{\an'}$. Hence $\an=\an'$.
\qed
\end{proof}
 
\subsection{Subclasses of a\-cyc\-lic nets}

The next definition introduces a\-cyc\-lic nets which will
represent individual causal histories. 

\begin{definition}[occurrence net]
\label{def:3}
    An \emph{occurrence net} is an a\-cyc\-lic net such that
    $|\pre{p}|\leq 1$ and $|\post{p}|\leq 1$, for every
    place $p$.
    \\
    Notation:
    $\ON$ is the set of all occurrence nets.
\EOD
\end{definition}
Occurrence nets  exhibit both backward de\-ter\-mi\-nism
and forward
de\-ter\-mi\-nism.
We also consider a\-cyc\-lic nets where only the
forward non-de\-ter\-mi\-nism
is allowed.
 
\begin{definition}[backward de\-ter\-mi\-nistic a\-cyc\-lic net]
\label{def:4}
    A \emph{backward de\-ter\-mi\-nistic a\-cyc\-lic net} is an a\-cyc\-lic net such that
    $|\pre{p}|\leq 1$, for every place $p$.
    \\
    Notation:
    $\BDAN$ is he set of all backward de\-ter\-mi\-nistic a\-cyc\-lic nets.
\EOD
\end{definition}

In the literature, backward de\-ter\-mi\-nistic a\-cyc\-lic nets
are sometimes called \emph{nonde\-ter\-mi\-nistic occurrence nets} or
even \emph{occurrence nets} (in which case occurrence nets
as defined above
are called \emph{de\-ter\-mi\-nistic occurrence nets}).

\begin{example}
\label{ex-1cccccccc}
In Figure~\ref{fig-1},
$\bdan_1$ is a  backward de\-ter\-mi\-nistic a\-cyc\-lic net, while
$\on_1$ and $\on_2$ are both occurrence nets.
\EOD
\end{example}

\begin{proposition}
\label{prop:2}
    $\ON\subset\BDAN\subset\AN$.
\end{proposition}
\begin{proof}
    The inclusions $\ON\subseteq\BDAN\subseteq\AN$ follow directly
    from the definitions.
    Moreover, in Figure~\ref{fig-1},
    $\an_1\in\AN\setminus\BDAN$ and
    $\bdan_1\in\BDAN\setminus\ON$.
\qed
\end{proof}
 
\begin{figure}[ht]
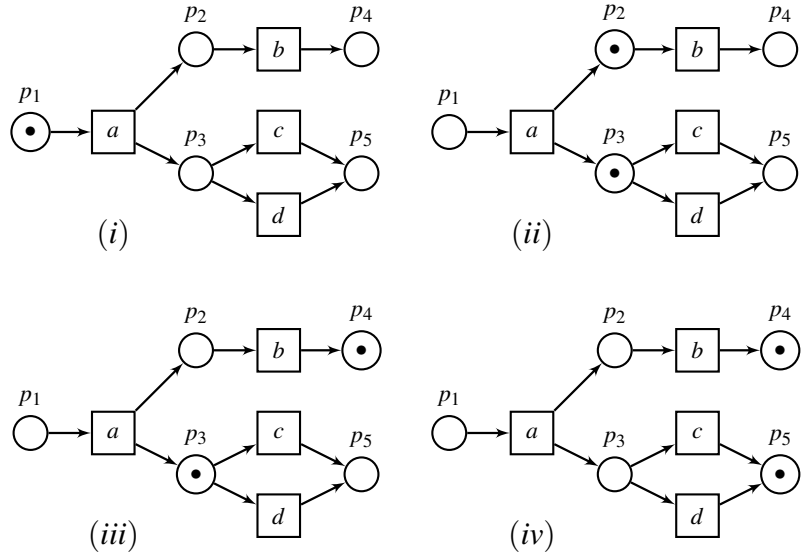

\begin{center}
\begin{tabular}{l@{~~~~~~~}l}

\StandardNet[0.55]
    \PlacN{P1}{ 0}{2}{\bullet}{p_1}
    \PlacN{P2}{ 4}{4}{}{p_2}
    \PlacN{P3}{ 4}{1}{}{p_3}
    \PlacN{P4}{ 8}{4}{}{p_4}
    \PlacN{P5}{ 8}{1}{}{p_5}

    \Whitetran{A}{ 2}{2}{a}
    \Whitetran{B}{ 6}{4}{b}
    \Whitetran{C}{ 6}{2}{c}
    \Whitetran{D}{ 6}{0}{d}

    \diredge{P1}{A}
    \diredge{P2}{B}
    \diredge{P3}{C}
    \diredge{P3}{D}

    \diredge{A}{P2}
    \diredge{A}{P3}
    \diredge{B}{P4}
    \diredge{C}{P5}
    \diredge{D}{P5} 
\Put{2}{-.5}{{\large$(i)$}}
\end{tikzpicture}
&
\StandardNet[0.55]
    \PlacN{P1}{ 0}{2}{}{p_1}
    \PlacN{P2}{ 4}{4}{\bullet}{p_2}
    \PlacN{P3}{ 4}{1}{\bullet}{p_3}
    \PlacN{P4}{ 8}{4}{}{p_4}
    \PlacN{P5}{ 8}{1}{}{p_5}

    \Whitetran{A}{ 2}{2}{a}
    \Whitetran{B}{ 6}{4}{b}
    \Whitetran{C}{ 6}{2}{c}
    \Whitetran{D}{ 6}{0}{d}

    \diredge{P1}{A}
    \diredge{P2}{B}
    \diredge{P3}{C}
    \diredge{P3}{D}

    \diredge{A}{P2}
    \diredge{A}{P3}
    \diredge{B}{P4}
    \diredge{C}{P5}
    \diredge{D}{P5} 
\Put{2}{-.5}{{\large$(ii)$}}
\end{tikzpicture}
\\[4mm]
\StandardNet[0.55]
    \PlacN{P1}{ 0}{2}{}{p_1}
    \PlacN{P2}{ 4}{4}{}{p_2}
    \PlacN{P3}{ 4}{1}{\bullet}{p_3}
    \PlacN{P4}{ 8}{4}{\bullet}{p_4}
    \PlacN{P5}{ 8}{1}{}{p_5}

    \Whitetran{A}{ 2}{2}{a}
    \Whitetran{B}{ 6}{4}{b}
    \Whitetran{C}{ 6}{2}{c}
    \Whitetran{D}{ 6}{0}{d}

    \diredge{P1}{A}
    \diredge{P2}{B}
    \diredge{P3}{C}
    \diredge{P3}{D}

    \diredge{A}{P2}
    \diredge{A}{P3}
    \diredge{B}{P4}
    \diredge{C}{P5}
    \diredge{D}{P5} 
\Put{2}{-.5}{{\large$(iii)$}}
\end{tikzpicture}
&
\StandardNet[0.55]
    \PlacN{P1}{ 0}{2}{}{p_1}
    \PlacN{P2}{ 4}{4}{}{p_2}
    \PlacN{P3}{ 4}{1}{}{p_3}
    \PlacN{P4}{ 8}{4}{\bullet}{p_4}
    \PlacN{P5}{ 8}{1}{\bullet}{p_5}

    \Whitetran{A}{ 2}{2}{a}
    \Whitetran{B}{ 6}{4}{b}
    \Whitetran{C}{ 6}{2}{c}
    \Whitetran{D}{ 6}{0}{d}

    \diredge{P1}{A}
    \diredge{P2}{B}
    \diredge{P3}{C}
    \diredge{P3}{D}

    \diredge{A}{P2}
    \diredge{A}{P3}
    \diredge{B}{P4}
    \diredge{C}{P5}
    \diredge{D}{P5} 
\Put{2}{-.5}{{\large$(iv)$}}
\end{tikzpicture}
\end{tabular}
\end{center}
\caption{Executing mixed step sequence 
$\{p_1\} \STEP{\{a\}}{}\{p_2,p_3\}\STEP{\{b\}}{}\{p_4,p_3\}
\STEP{\{c\}}{}\{p_4,p_5\}$
and showing the consecutive snapshots $(i)\to(ii)\to(iii)\to(iv)$.}
\label{fig-1cccvvvvvvvv}
\end{figure}

\begin{figure}[ht]
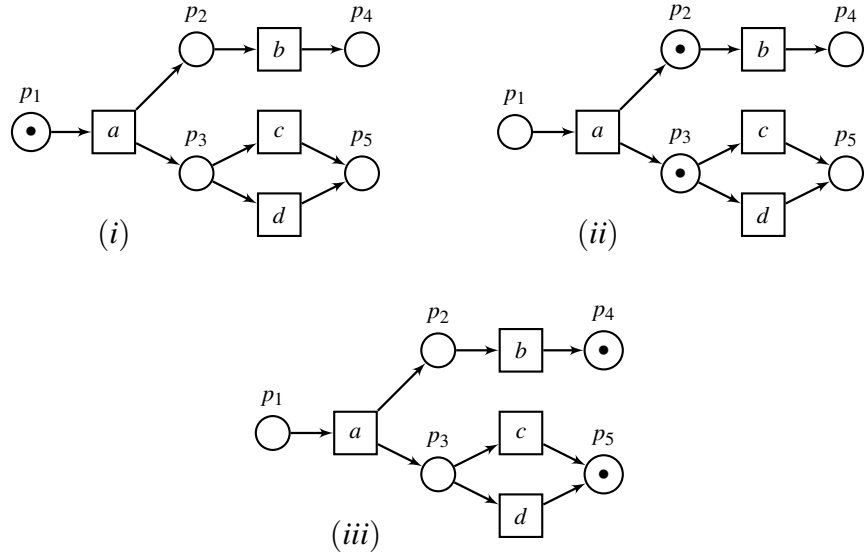

\begin{center}
\begin{tabular}{c}

\StandardNet[0.55]
    \PlacN{P1}{ 0}{2}{\bullet}{p_1}
    \PlacN{P2}{ 4}{4}{}{p_2}
    \PlacN{P3}{ 4}{1}{}{p_3}
    \PlacN{P4}{ 8}{4}{}{p_4}
    \PlacN{P5}{ 8}{1}{}{p_5}

    \Whitetran{A}{ 2}{2}{a}
    \Whitetran{B}{ 6}{4}{b}
    \Whitetran{C}{ 6}{2}{c}
    \Whitetran{D}{ 6}{0}{d}

    \diredge{P1}{A}
    \diredge{P2}{B}
    \diredge{P3}{C}
    \diredge{P3}{D}

    \diredge{A}{P2}
    \diredge{A}{P3}
    \diredge{B}{P4}
    \diredge{C}{P5}
    \diredge{D}{P5} 
\Put{2}{-.5}{{\large$(i)$}}
\end{tikzpicture}
~~~~~~~~~~~~~~~~
\StandardNet[0.55]
    \PlacN{P1}{ 0}{2}{}{p_1}
    \PlacN{P2}{ 4}{4}{\bullet}{p_2}
    \PlacN{P3}{ 4}{1}{\bullet}{p_3}
    \PlacN{P4}{ 8}{4}{}{p_4}
    \PlacN{P5}{ 8}{1}{}{p_5}

    \Whitetran{A}{ 2}{2}{a}
    \Whitetran{B}{ 6}{4}{b}
    \Whitetran{C}{ 6}{2}{c}
    \Whitetran{D}{ 6}{0}{d}

    \diredge{P1}{A}
    \diredge{P2}{B}
    \diredge{P3}{C}
    \diredge{P3}{D}

    \diredge{A}{P2}
    \diredge{A}{P3}
    \diredge{B}{P4}
    \diredge{C}{P5}
    \diredge{D}{P5} 
\Put{2}{-.5}{{\large$(ii)$}}
\end{tikzpicture}
\\[4mm]
\StandardNet[0.55]
    \PlacN{P1}{ 0}{2}{}{p_1}
    \PlacN{P2}{ 4}{4}{}{p_2}
    \PlacN{P3}{ 4}{1}{}{p_3}
    \PlacN{P4}{ 8}{4}{\bullet}{p_4}
    \PlacN{P5}{ 8}{1}{\bullet}{p_5}

    \Whitetran{A}{ 2}{2}{a}
    \Whitetran{B}{ 6}{4}{b}
    \Whitetran{C}{ 6}{2}{c}
    \Whitetran{D}{ 6}{0}{d}

    \diredge{P1}{A}
    \diredge{P2}{B}
    \diredge{P3}{C}
    \diredge{P3}{D}

    \diredge{A}{P2}
    \diredge{A}{P3}
    \diredge{B}{P4}
    \diredge{C}{P5}
    \diredge{D}{P5} 
\Put{2}{-.5}{{\large$(iii)$}}
\end{tikzpicture} 
\end{tabular}
\end{center}
\caption{Executing mixed step sequence 
$\{p_1\} \STEP{\{a\}}{}\{p_2,p_3\}\STEP{\{b,c\}}{}\{p_4,p_5\}$
and showing the consecutive snapshots $(i)\to(ii)\to(iii)$.
}
\label{fig-1cccvvvvvvvvvvvvvvvvvvvvvvv} 
\end{figure}

\section{Step sequence semantics of a\-cyc\-lic nets}

We now   introduce notions related to the  behaviour
of a\-cyc\-lic nets.
The intuition behind them is as follows:
\begin{description}
\item[marking]{\ }
\\
A global state of a possible execution history
of the system modelled by a net.

\item[initial and final markings]{\ }
\\
There is always a single initial global state. 
In general, there is more than one final global state, which
corresponds to the fact that a single net may capture
alternative execution histories.

\item[step]{\ }
\\
A
set of events which might have occurred simultaneously and effected
a move from one global state to another.

\item[enabled step]{\ }
\\
A step which can be executed at a global state
thanks to all its input places being present/marked.
 
\item[mixed step sequence]{\ }
\\
An alternating sequence of global states and
executed steps transforming one global state into another.
 
\item[reachable marking]{\ }
\\
A global state which can be obtained by executing a sequence of
steps starting from the initial marking.

\item[scenario]{\ }
\\
An occurrence net  providing a representation of
a  system history.

\item[well-formed a\-cyc\-lic net]{\ }
\\
An a\-cyc\-lic net where each execution history 
from the initial state has an unambiguous interpretation in terms of 
causality and concurrency.
\end{description}

\begin{definition}[step and marking of a\-cyc\-lic net]
\label{def:5}
    Let $\an$ be an a\-cyc\-lic net.
\begin{enumerate}
\item
    $\STEPS(\an)\triangleq \{U\in \pset{T_\an}\setminus\{\es\}\mid
      \forall t\neq u\in U: \pre{t}\cap\pre{u}=\es\}$
    are the \emph{steps}.
\item
    $\MARK(\an)\triangleq \pset{P_\an}$ are the \emph{markings}.
\item
    $M_\an^\init\triangleq P_\an^\init$
    is the \emph{initial} marking.
\EOD
\end{enumerate}
\end{definition}
Graphically, markings are indicated by black tokens placed inside the 
corresponding circles.

Note that in a step no two transitions can share a pre-place.

\begin{example}
\label{ex-2s}
For
the a\-cyc\-lic nets $\an_1$ and $\bdan_1$ depicted in Figure~\ref{fig-1}($a,b$),
we have
$\STEPS(\an_1)=\STEPS(\bdan_1)=
 \{U\in \pset{\{a,b,c,d\}}\setminus\{\es\}\mid
 c\notin U\vee d\notin U\}$ and
 $M_{\an_1}^\init=M_{\bdan_1}^\init=\{p_1\}$.
\EOD
\end{example}

\begin{definition}[enabled and executed step of acyclic net]
\label{def:6}
    Let $M$ be a marking of  an a\-cyc\-lic net $\an$.
\begin{enumerate}
\item
    $\enabled_\an(M)\triangleq\{U\in\STEPS(\an)\mid \pre{U}\subseteq M\}$
    are the steps \emph{enabled} at $M$.
\item
    A step $U\in\enabled_\an(M)$
    can be \emph{executed} and yield
    $M'\triangleq (M\cup\post{U})\setminus\pre{U}$.
    \\
    Notation:  $M\STEP{U}{\an}M'$.
\EOD
\end{enumerate}
\end{definition}
Enabling a step in a
global state amounts to having all its pre-places marked.
The execution of such a step adds tokens to all its post-places and then
removes tokens from all its pre-places.
Note that in the standard Petri net semantics the execution
of a step leads to $M'= (M\setminus\pre{U})\cup\post{U}$,
which yields a different result unless $\pre{U}\cap\post{U}=\es$.
However, the latter always holds when the acyclic net is well-formed,
as defined later in this paper. 

\begin{definition}[mixed step sequence and step sequence of a\-cyc\-lic net]
\label{def:7aee}
    Let $M_0,\dots,M_k$ ($k\geq 0$) be
    markings and $U_1,\dots,U_k$ be steps
    of an a\-cyc\-lic net $\an$ such that we have
    $M_{i-1}\STEP{U_i}{\an}M_i$,
    for every $1\leq i\leq k$.
\begin{enumerate}
\item
    $\mu=M_0U_1M_1\dots M_{k-1}U_k M_k$
    is a \emph{mixed step sequence from $M_0$ to $M_k$}.

\item
    $\sigma=U_1\dots U_k$ is
    a \emph{step sequence from $M_0$ to $M_k$}.
\end{enumerate}
    The above two notions are denoted by $M_0\mixSTEP{\mu}{\an}M_k$
    and $M_0\STEP{\sigma}{\an}M_k$, respectively.
    Moreover,
    $M_0\STEP{\sigma}{\an}$ denotes
    that $\sigma$ is a step sequence \emph{enabled $M_0$}, and
    $M_0\STEP{}{\an}M_k$
    denotes that $M_k$ is \emph{reachable from $M_0$}.
\EOD
\end{definition}
If $k=0$ then $\mu=M_0$ and the corresponding
step sequence $\sigma$ is the \emph{empty} sequence denoted by $\lambda$.

In the last definition, the starting point of an execution
is an arbitrary marking.
The next definition introduces a number of
be\-ha\-vioural notions, assuming in each case that the starting point
of system executions is
the default initial marking.

\begin{definition}[behaviour of a\-cyc\-lic net]
\label{def:7rtt}
    The following sets capture various be\-ha\-vioural notions
    related to step sequences and reachable markings
    of an a\-cyc\-lic net $\an$.
\begin{enumerate}
\item
    $\sseq(\an)\triangleq\{\sigma\mid  M_\an^\init\STEP{\sigma}{\an}M\}$
    \hfill \emph{step sequences}.

\item
    $\mixsseq(\an)\triangleq\{\mu\mid  M_\an^\init\mixSTEP{\mu}{\an}M\}$
    \hfill   \emph{mixed step sequences}.

\item
    $\maxsseq(\an)\triangleq\{\sigma\in\sseq(\an)\mid
         \neg\exists U: \sigma U\in\sseq(\an)\}$
    \\\hspace*{\fill}   \emph{maximal step sequences}.

\item
    $\maxmixsseq(\an)\triangleq\{\mu\in\mixsseq(\an)\mid
         \neg\exists U,M: \mu UM\in\mixsseq(\an)\}$
    \\\hspace*{\fill} \emph{maximal mixed step sequences}.

\item
    $\reach(\an)\triangleq \{M\mid
               M_\an^\init\STEP{ }{\an}M\}$
    \hfill   \emph{reachable markings}.

\item
    $\finreach(\an)\triangleq\{M\mid \exists \sigma\in\maxsseq(\an):
               M_\an^\init\STEP{\sigma}{\an}M\}$
    \\\hspace*{\fill}  \emph{final reachable markings}.
\item
    $\fseq(\an)=\{U_1\dots U_k\in\sseq(\an)\mid
    k\geq 1\implies
    |U_1|=\cdots=|U_k|=1\}$
    \\\hspace*{\fill}  \emph{firing sequences}.
\end{enumerate}
\EOD
\end{definition}

    We can treat individual  transitions as singleton steps;
    for instance, a step sequence $\{t\}\{u\}\{w,v\}\{z\}$
    can be denoted by $tu\{w,v\}z$.

\begin{example}
\label{ex-5j}
    The following hold for the
    a\-cyc\-lic net in Figure~\ref{fig-1}($b$).
\begin{enumerate}
\item
    $\sseq(\bdan_1)=\{\lambda,a,ab,ac,ad,abc,acb,abd,adb,a\{b,c\},
     a\{b,d\}\}$.

\item
    $\mixsseq(\bdan_1)=\{
    \{p_1\}, \{p_1\}a\{p_2,p_3\},\{p_1\}a\{p_2,p_3\}\{b,c\}\{p_4,p_5\},
    \dots\}$.

\item
    $\maxsseq(\bdan_1)=\{
     abc,acb,a\{b,c\},
     abd,adb,a\{b,d\}\}$.

\item
    $\maxmixsseq(\bdan_1)=\{
    \{p_1\}a\{p_2,p_3\}\{b,c\}\{p_4,p_5\},\dots\}$.

\item
    $\reach(\bdan_1)=\{\{p_1\}, \{p_2,p_3\},\{p_2,p_5\},\{p_2,p_6\},
    \{p_4,p_3\},\dots\}$.

\item
    $\finreach(\bdan_1)=\{
    \{p_4,p_5\},\{p_4,p_6\}\}$.
\item
    $\fseq(\bdan_1)=\{\lambda,a,ab,ac,ad,abc,acb,abd,adb\}$.
\end{enumerate}
Moreover, Figures~\ref{fig-1cccvvvvvvvv}
and~\ref{fig-1cccvvvvvvvvvvvvvvvvvvvvvvv} show consecutive 
snapshots of acyclic nets involved in mixed step sequences.
\EOD
\end{example}

The next result shows that it is always possible to arbitrarily serialise 
an enabled step and, as a result, relate the step sequence based semantics and 
the firing
sequence based semantics.
\begin{proposition}
\label{pr-hhhd}
    Let $\an$ be an a\-cyc\-lic net and $U=\{t_1,\dots, t_k\}$
    be a step enabled at a marking $M$.
\begin{enumerate}
\item
    If $U'$ and $U''$ form a
    partition of $U$, then
    $M\STEP{U'U''}{\an}$.
\item
    $M\STEP{\{t_1\}\dots\{t_k\}}{\an}$.
\end{enumerate}
\end{proposition}
\begin{proof}
(1)
    Clearly, $U'$ is a step enabled at $M$.
    Let $M\STEP{U'}{\an}M'$.
    We observe
    that $U''$ is enabled at $M'$ which follows from
    $\pre{U}\cup\pre{U''}\subseteq M$
    (see Definition~\ref{def:6}(1)),
    $\pre{U'}\cap\pre{U''}=\es$
    (see Definition~\ref{def:5}(1)), and
    Definition~\ref{def:6}(2).
    Hence
    $M\STEP{U'U''}{\an}$.

(2)
    As  for $k=1$ there is nothing to show,
    suppose that $k>1$.
    Then $U'=\{t_1\}$ and $U''=\{t_2,\dots, t_k\}$
    form a partition of $U$.
    Hence, by part (1),
    $M\STEP{\{t_1\}\{t_2,\dots, t_k\}}{\an}$.
    By repeating the same argument
    $k-1$ times, we obtain  $M\STEP{\{t_1\}\dots\{t_k\}}{\an}$.
\qed
\end{proof}

\subsection{Behaviour of subnets of a\-cyc\-lic nets}

The structural inclusion of a\-cyc\-lic
nets sharing the initial marking implies inclusion of the behaviours
they generate.

\begin{proposition}
\label{pr-uususs}
    Let $\an\sqsubseteq\an'$ be a\-cyc\-lic nets.
    Then $f(\an) \subseteq f(\an')$, for
    $f=\mixsseq,\sseq,\reach,\fseq$.
\end{proposition}
\begin{proof}
    It suffices to show that $\mixsseq(\an) \subseteq \mixsseq(\an')$.

    Let $\mu=M_0U_1M_1\dots M_{k-1}U_k M_k\in\mixsseq(\an)$.
    Then $\mu\in\mixsseq(\an')$ which
    follows from $M_\an^\init=M_0=M_{\an'}^\init$
    (see Definition~\ref{def:4aaa}) and, for every $0\leq i<k$,
    we have
    $M_i\STEP{U_{i+1}}{{\an'}}M_{i+1}$
    which follows from $M_i\STEP{U_{i+1}}{\an}M_{i+1}$ and
    Definition~\ref{def:4adddaa}.
\qed
\end{proof} 

Due to being both backward-deterministic and
forward-deterministic, occurrence nets exhibit
additional useful be\-ha\-vioural properties.

\begin{proposition}
\label{pr-ueeu}
    Let $\sigma=U_1\dots U_k$ ($k\geq 0$) be a
    sequence of nonempty sets of transitions of 
    an occurrence net $\on$, and 
    $M$ be a reachable marking of $\on$.
    Moreover, let $T=\bigcup\sigma$.
\begin{enumerate}
\item
    $\sigma\in\sseq(\on)$ iff
    for every $1\leq i\leq k$, 
    $U_i\cap (U_1\cup\dots \cup U_{i-1})=\es$ and
\[
    \pre{(\pre{U_i})}\subseteq U_1\cup\dots\cup U_{i-1}\;,
\]
    where $U_1\cup\dots \cup U_{i-1}=\es$ for $i=1$.
\item
    $\sigma\in\sseq(\on)$
    implies
    $M_\on^\init\STEP{\sigma}{\on}
      (M_\on^\init\cup \post{T})
        \setminus \pre{T}$.
\item
    $T_\on=\bigcup\{\bigcup\xi\mid \xi\in\sseq(\on)\}$.
\item
    $M\STEP{U}{\on}M'$ implies $U'\in\enabled_\on(M')$, for every
    $U'\in\enabled_\on(M)$ such that $U\cap U'=\es$.
\item
    $\maxsseq(\on)=\{\xi\in\sseq(\on)\mid \bigcup\xi=T_\on\}$.
\item
    $\finreach(\on)=\{P^\fin_\on\}$.
\item
    If  $t\in\enabled_\on(M)$,
    then $t\in\enabled_\on(M')$ and 
    $M^\init_\on\STEP{}{\on}M'\STEP{}{\on}M$, where  
    $M'=(P_\on^\init \cup \post{V})\setminus \pre{V}$ for
    $V=\{u\in T_\on\mid u F^+ t\}$.
\end{enumerate}
\end{proposition}
\begin{proof}
     If follows from the standard properties of occurrence nets.
\qed
\end{proof}
That is, an occurrence net is de\-ter\-mi\-nistic in the sense that
each transition has a fixed set of direct or indirect
predecessors which have to occur first
before the transition is executed
(see Proposition~\ref{pr-ueeu}(1,7)), and once it is enabled
no other transition can disable it (see Proposition~\ref{pr-ueeu}(4)).
Moreover, all the transitions can be executed (see Proposition~\ref{pr-ueeu}(3)),
the order in which transitions are executed does not
influence the resulting marking (see Proposition~\ref{pr-ueeu}(2)),
the maximal step sequences are those which use all the transitions, and
there is exactly one final marking (see Proposition~\ref{pr-ueeu}(6)).

\section{Causality in a\-cyc\-lic nets}

There is a straightforward way of
introducing  causality in a\-cyc\-lic nets
by looking at their \emph{step sequences}.
The idea can be explained as follows.

\begin{quote}
Let $t$ and $u$ be two transitions of an a\-cyc\-lic net $\an$.
Then
$u$ is a \emph{cause} of $t$ if, for every step sequence $\sigma U\in\sseq(\an)$
such that $t\in U$, it is the case that $u$ occurs in $\sigma$.
\end{quote}
One can check that if $u$ is a cause of $t$ in the above
sense, then  $uF^+_\an t$. In other words,
if $u$ is a cause of $t$ then there must be a direct
path from $u$ to $t$ in the graph
of $\an$. The converse, however, does not hold.
For instance, Figure~\ref{fig-2}($a$) depicts an a\-cyc\-lic net
in which there is a directed path from $a$ to $c$,
but
there is also a step sequence $\{b\}\{c\}$ in which $c$ is not preceded
by an occurrence of $a$.

However, when $\an$ is
a backward de\-ter\-mi\-nistic a\-cyc\-lic net, then
$uF^+_\an t$ implies that $u$ is a cause of $t$.
This is often referred to as backward
de\-ter\-mi\-nism.
Forward de\-ter\-mi\-nism is
an  additional property enjoyed by occurrence nets.
It means, for example, that an enabled transition
cannot be disabled by executing another transition.
Also,
after executing an arbitrary step sequence, it is possible to
continue the execution until all the
transitions have been executed.

A conclusion of the above short discussion is that
causality can be investigated by using purely
graph theoretic concepts (in this case, directed paths in
the graphs of a\-cyc\-lic nets).

\begin{figure}[ht]
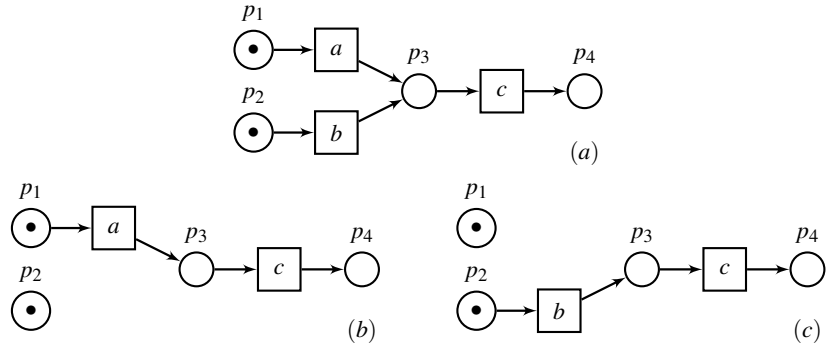

\begin{center}
\StandardNet[0.55]
    \PlacN{P1}{ 0}{2}{\bullet}{p_1}
    \PlacN{P2}{ 0}{0}{\bullet}{p_2}
    \PlacN{P3}{ 4}{1}{}{p_3}
    \PlacN{P4}{ 8}{1}{}{p_4}

    \Whitetran{A}{ 2}{2}{a}
    \Whitetran{B}{ 2}{0}{b}
    \Whitetran{C}{ 6}{1}{c}

    \diredge{P1}{A}
    \diredge{P2}{B}
    \diredge{P3}{C}

    \diredge{A}{P3}
    \diredge{B}{P3}
    \diredge{C}{P4}
\Put{8}{-.5}{$(a)$}
\end{tikzpicture}

\StandardNet[0.55]
    \PlacN{P1}{ 0}{2}{\bullet}{p_1}
    \PlacN{P2}{ 0}{0}{\bullet}{p_2}
    \PlacN{P3}{ 4}{1}{}{p_3}
    \PlacN{P4}{ 8}{1}{}{p_4}

    \Whitetran{A}{ 2}{2}{a}
    \Whitetran{C}{ 6}{1}{c}

    \diredge{P1}{A}
    \diredge{P3}{C}

    \diredge{A}{P3}
    \diredge{C}{P4}
\Put{8}{-.5}{$(b)$}
\end{tikzpicture}
~~~~~~~~~
\StandardNet[0.55]
    \PlacN{P1}{ 0}{2}{\bullet}{p_1}
    \PlacN{P2}{ 0}{0}{\bullet}{p_2}
    \PlacN{P3}{ 4}{1}{}{p_3}
    \PlacN{P4}{ 8}{1}{}{p_4}

    \Whitetran{B}{ 2}{0}{b}
    \Whitetran{C}{ 6}{1}{c}

    \diredge{P2}{B}
    \diredge{P3}{C}

    \diredge{B}{P3}
    \diredge{C}{P4}
\Put{8}{-.5}{$(c)$}
\end{tikzpicture}

\end{center}
\caption{An a\-cyc\-lic net $(a)$ and its two maximal scenarios $(b,c)$,
with their initial markings.}
\label{fig-2}
\end{figure}

\section{Well-formed a\-cyc\-lic nets}

A fundamental consistency criterion applied to a\-cyc\-lic nets
is well-formedness. Its essence is to
ensure an unambiguous representation of causality in
behaviours they represent.
The definition of a well-formed a\-cyc\-lic
net is derived from the notion of a well-formed step
sequence.

\begin{definition}[well-formed step sequence of a\-cyc\-lic net]
\label{def:10aa}
    A step sequence $U_1\dots U_k$ of
    an a\-cyc\-lic net is \emph{well-formed}
    if the following hold:
\begin{enumerate}
\item
    $\post{t}\cap\post{u}=\es$,
    for every $1\leq i\leq k$ and all $t\neq u\in U_i$.
\item
    $\post{U_i}\cap\post{U_j}=\es$,
    for all $1\leq i<j\leq k$.
\EOD
\end{enumerate}
\end{definition}
In a well-formed step sequence of an
a\-cyc\-lic net, no place receives a token more than once.
It then follows, eg.g., that in
such a step sequence no place is a pre-place   
of an executed step more than once,
the order of execution of transitions does not influence the
resulting marking, and each step sequence can be sequentialised
to a firing sequence.

\begin{proposition}
\label{pr-jeff}
    Let $\sigma=U_1\dots U_k$ be a well-formed
    step sequence of an a\-cyc\-lic net $\an$, and
    $M_\an^\init=M_0 U_1 M_1 \dots M_{k-1} U_k M_k$ be the corresponding 
    mixed step sequence.
    Moreover, let $T=\bigcup\sigma$ and
    $U_i=\{t_i^1,\dots,t_i^{m_i}\}$, for every
    $1\leq i\leq k$.
\begin{enumerate}
\item
    $\pre{U_i}\cap\post{U_i}=\es$,
    for every $1\leq i\leq k$.

\item
    $\pre{U_i}\cap\pre{U_j}=\es$ and $U_i\cap U_j=\es$,
    for all $1\leq i<j\leq k$.

\item
    $M_k = M_\an^\init\cup \post{T}\setminus \pre{T}$.

\item
    $t_1^1\dots t_1^{m_1}\dots t_k^1 \dots t_k^{m_k}\in\fseq(\an)$.
    
\item 
    $\enabled_\an(M_i)\cap\enabled_\an(M_m)\subseteq\enabled_\an(M_j)$,
    for all $1\leq i<j<m\leq n$. 
\end{enumerate}
\end{proposition}

\begin{proof}
(1)
    Suppose that $p\in\pre{U}_i\cap\post{U}_i$.
    Then $p\notin M^\init_\an$, and so there is 
    $j<i$ such that $p\in\post{U}_j$, contradicting 
    Definition~\ref{def:10aa}(2).

(2)
    $U_i\cap U_j=\es$ follows from Definition~\ref{def:10aa}(2).
    \\
    Suppose that $p\in\pre{U}_i\cap\pre{U}_j$.
    Then, by Definition~\ref{def:6}(2) and part (1),
    $p\in M_{i-1}\cap M_{j-1}$ and $p\notin M_i$.
    Thus $p\notin M_0=M^\init_\an$ and
    there are
    $t\in U_1\cup\dots \cup U_{i-1}$
    and
    $u\in U_{i+1}\cup\dots \cup U_{j-1}$
    such that $p\in\post{t}\cap\post{u}$, contradicting 
    Definition~\ref{def:10aa}(2).

(3)
    If follows from the fact that, for a given
    place $p$, at most one transition $t_i^j$ belongs to $\pre{p}$
    (see Definition~\ref{def:10aa}),
    at most one transition $t_i^j$ belongs to $\post{p}$
    (see part (2) and Definition~\ref{def:5}(1)), 
    and if $p\notin M_\an^\init$
    then there must exist $t_l^m\in\pre{p}$ with $l<j$.

(4)
    Let  $\tau_i=t_i^1\dots t_i^{m_i}$, for every $1\leq i\leq k$.
    It suffices to observe that if $0\leq i<k$
    and $\tau_0\dots\tau_i \in\fseq(\an)$, where
    $\tau_0=\lambda$,
    then $\tau_0\dots\tau_{i+1} \in\fseq(\an)$.
    Indeed, since $\sigma$ is well-formed,
    so is $\tau_0\dots\tau_i$. Hence, by
    part (3), $M_0\STEP{\tau_0\dots\tau_i}{\an}M_i$.
    Thus, by Proposition~\ref{pr-hhhd}(2),
    $\tau_0\dots\tau_{i+1} \in\fseq(\an)$.
    
(5)
    Suppose that  
    $t\in \enabled_\an(M_i)\cap\enabled_\an(M_m)\setminus\enabled_\an(M_j)$.
    Then there is $p\in \pre{t}$
    such that $p\in M_i\cap M_m\setminus M_j$. Hence $p\notin P_\an^\init$,
    $p\in \post{(U_1\cup\dots\cup U_i)}$
    and $p\in \post{(U_{j+1}\cup\dots\cup U_m)}$, contradicting  
    Definition~\ref{def:10aa}(2).
\qed
\end{proof}

To develop a sound treatment of
causality in an a\-cyc\-lic net, it is  required that
all step sequences are well-formed.
Moreover, it is required that a well-formed a\-cyc\-lic
net has no redundant transitions which can never be executed
from the initial marking.

\begin{definition}[well-formed a\-cyc\-lic net]
\label{def:10}
    An a\-cyc\-lic net  is \emph{well-formed}
    if  each transition occurs in at least one
    step sequence and
    all the step sequences are well-formed.
    \\
    Notation: $\WFAN$ is the set of all well-formed a\-cyc\-lic nets.
\EOD
\end{definition}

\begin{example}
\label{ex-6}
The a\-cyc\-lic nets in Figure~\ref{fig-1}($a,b$)
are well-formed, but
the a\-cyc\-lic net in Figure~\ref{fig-2}($a$)
is not.
\\
The reason why the a\-cyc\-lic net in Figure~\ref{fig-2}($a$)
is not well-formed can be explained in terms of \emph{OR-causality}
which it exhibits. Intuitively, one can
execute $a$ and $b$ first, and then $c$. It is, however,
impossible to state in such a case whether $c$ was caused by $a$
or by $b$ (we know, by looking at Figure~\ref{fig-2}($b,c$)
that only one of $a$ and $b$ is needed to cause $c$).
In the context of an investigation, one would
presumably attempt to `repair' the a\-cyc\-lic net in Figure~\ref{fig-2}($a$)
to resolve the problem. Two possible outcomes of such an attempt
are depicted in Figure~\ref{fig-2dd}.
\EOD
\end{example}

Checking that an a\-cyc\-lic net is well-formed 
can be done by looking just at its firing sequences.

\begin{proposition}
\label{prop:10sdxx}
    An a\-cyc\-lic net $\an$ is well-formed iff
    each transition occurs in at least one
    firing sequence and
    all the firing sequences
    are well-formed.
\end{proposition}
\begin{proof}
$(\Longrightarrow)$
    Since $\fseq(\an)\subseteq\sseq(\an)$,
    it suffices to show that each transition
    occurs in at least one firing sequence.
    This, however, follows from
    $\an$ being well-formed and
    and Proposition~\ref{pr-jeff}(4).

$(\Longleftarrow)$
    Since $\fseq(\an)\subseteq\sseq(\an)$, it suffices to show that
    each step sequence is well-formed.
    If this is not the case, then there is
    $U_1\dots U_k\in\sseq(\an)$ ($k\geq 1$) such that
    $U_1\dots U_{k-1}$ is well-formed and
    $U_1\dots U_k$ is not well-formed.
    Let
    $U_i=\{t_i^1,\dots,t_i^{m_i}\}$, for every
    $1\leq i\leq k$.
    Then, by Propositions~\ref{pr-hhhd}(2)
    and~\ref{pr-jeff}(3,4),
    $\tau=t_1^1\dots t_1^{m_1}\dots t_k^1\dots t_k^{m_k}\in\fseq(\an)$.
    Hence, since $\tau$ is well-formed, $\sigma$
    is also well-formed, yielding a contradiction.
\qed
\end{proof} 

Well-formedness is ensured when backward non-determinism is
not allowed.

\begin{proposition}
\label{prop:3xx}
    All step sequences of a backward deterministic
    a\-cyc\-lic net are well-formed.
\end{proposition}
\begin{proof}
    If the result does not hold then, by Propositions~\ref{pr-hhhd}(2)
    and~\ref{pr-jeff}(4),
    there is $t_1\dots t_kt \in\fseq(\an)$ ($k\geq 1$)
    such that $t_1\dots t_k$ is well-formed and $t_1\dots t_kt$
    is not well-formed. Hence there is $i\leq k$ such that
    $\post{t_i}\cap \post{t}\neq\es$. Take any $p\in \pre{t}$.
    Since $\an$ is backward deterministic, $t_i=t$.
    Hence $i<k$ and there is $i<j\leq k$ such that $p\in \post{t}_j$.
    As $p\notin M_\an^\init$ and $p\in \pre{t}_i$,
    there is $1\leq m<i$ such that $p\in \post{t}_m$.
    This, however, contradicts $t_1\dots t_k$ being well-formed.
\qed
\end{proof}
 
All occurrence nets
are well-formed.

\begin{proposition}
\label{prop:3}
    $\ON\subset\WFAN\subset\AN$.
\end{proposition}
\begin{proof}
    Clearly, $\WFAN\subseteq\AN$. Moreover,
    $\ON\subseteq\WFAN$ follows from
    $\ON\subseteq\BDAN$ and Propositions~\ref{pr-ueeu}(3)
    and~\ref{prop:3xx}.
    Moreover, $\bdan_1\in\WFAN\setminus\ON$
    in Figure~\ref{fig-2}, and
    the a\-cyc\-lic net in Figure~\ref{fig-2}($a$)
    is not well-formed.
\qed
\end{proof}

\begin{figure}[ht]
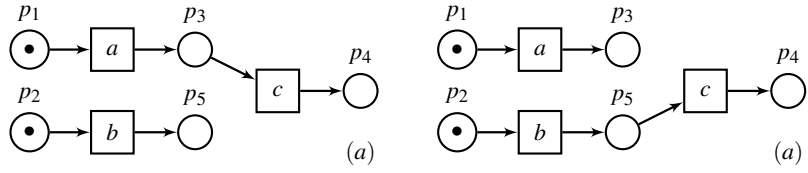

\begin{center}
\StandardNet[0.55]
    \PlacN{P1}{ 0}{2}{\bullet}{p_1}
    \PlacN{P2}{ 0}{0}{\bullet}{p_2}
    \PlacN{P3}{ 4}{2}{}{p_3}
    \PlacN{P4}{ 8}{1}{}{p_4}
    \PlacN{P5}{ 4}{0}{}{p_5}

    \Whitetran{A}{ 2}{2}{a}
    \Whitetran{B}{ 2}{0}{b}
    \Whitetran{C}{ 6}{1}{c}

    \diredge{P1}{A}
    \diredge{P2}{B}
    \diredge{P3}{C}

    \diredge{A}{P3}
    \diredge{B}{P5}
    \diredge{C}{P4}
\Put{8}{-.5}{$(a)$}
\end{tikzpicture}
~~~~~~
\StandardNet[0.55]
    \PlacN{P1}{ 0}{2}{\bullet}{p_1}
    \PlacN{P2}{ 0}{0}{\bullet}{p_2}
    \PlacN{P3}{ 4}{2}{}{p_3}
    \PlacN{P4}{ 8}{1}{}{p_4}
    \PlacN{P5}{ 4}{0}{}{p_5}

    \Whitetran{A}{ 2}{2}{a}
    \Whitetran{B}{ 2}{0}{b}
    \Whitetran{C}{ 6}{1}{c}

    \diredge{P1}{A}
    \diredge{P2}{B}
    \diredge{P5}{C}

    \diredge{A}{P3}
    \diredge{B}{P5}
    \diredge{C}{P4}
\Put{8}{-.5}{$(a)$}
\end{tikzpicture}
\end{center}
\caption{Well-formed a\-cyc\-lic nets.}
\label{fig-2dd}
\end{figure}

\section{Scenarios in a\-cyc\-lic nets}

An a\-cyc\-lic net may exhibit both forward and backward nondeterminism 
and, as a result, represent several different possible execution histories.
The role of the next notion is to identify structurally  all such execution 
histories which can then be inspected, simulated, and analysed. 

Scenarios of an a\-cyc\-lic net are   a\-cyc\-lic subnets which start at the same initial 
marking and are both backward
and forward de\-ter\-mi\-nistic. As a result, each scenario represents
a distinct execution history with clearly de\-ter\-mi\-ned causal
relationships. Scenarios are  much more abstract than (mixed) step sequences,
as one scenario will in general correspond to many step sequences.
 
\begin{definition}[scenario and maximal scenario of  a\-cyc\-lic net]
\label{def:11} 
    A \emph{scenario} of an  a\-cyc\-lic net $\an$  is an occurrence net $\on$ such that 
    $\on\sqsubseteq\an$.
    Moreover, $\on$ is \emph{maximal} if there is 
    no scenario $\on'\neq\on$ such that $\on\sqsubseteq\on'\sqsubseteq\an$.
\\ 
    Notation:
    $\scenarios(\an)$ and $\maxscenarios(\an)$ are respectively
    the sets of all scenarios and maximal scenarios
    of $\an$.
\EOD 
\end{definition}
Intuitively, scenarios  represent
possible executions (concurrent histories),
which are both deterministic and consistent with the 
dependencies implied by the flow relation. 
Maximal scenarios are complete in the sense that
they cannot be extended.

Although scenarios represent behaviours of a\-cyc\-lic nets
in a different way than (mixed) step sequences,
there is a close relationship between these two approaches.
In particular, scenarios do not generate executions
which are not generated by the original a\-cyc\-lic net.
\begin{proposition}
\label{prop:8}
    $\bigcup f(\scenarios(\an))\subseteq f(\an)$, for every
    a\-cyc\-lic net $\an$ and  $f=\mixsseq,\sseq, \reach,\fseq$. 
\end{proposition}
\begin{proof}
    It follows from Definition~\ref{def:11} and Proposition~\ref{pr-uususs}.
\qed
\end{proof}

The inclusions in the above result cannot be reversed.

\begin{example}
Consider 
the (non-well-formed) a\-cyc\-lic net $\an$ in Figure~\ref{fig-2}($a$)
which generates  step sequence $\sigma=\{a,b\}\{c\}$.
It has exactly two maximal scenarios shown in Figure~\ref{fig-2}($b,c$),
neither of which can execute $\sigma$.  
\EOD
\end{example}
 
Definition~\ref{def:11} introduced scenarios 
in a structural way.
An alternative is to do this be\-ha\-viourally,   
through well-formed step sequences.

\begin{definition}[scenario of well-formed step sequence of a\-cyc\-lic net]
\label{def:13}
    The  \emph{scenario  induced} by a well-formed step
    sequence $\sigma$ of an a\-cyc\-lic net $\an$ is the triple:
\[
\scenario_\an(\sigma)
    \triangleq
    (P,T,F_\an|_{(P\times T)\cup(T\times P)})
\]
   where $T=\bigcup\sigma$
  and $P=P_\an^\init\cup \POST{\an}{T}$.
\EOD
\end{definition}

\begin{example}
\label{ex-8}
In Figure~\ref{fig-1},
$\on_1=\scenario_{\an_1}(\{a\}\{b,c\})$.
Note that
$\on_1$ can also be derived as
$\scenario_{\an_1}(\{a\}\{b\}\{c\})$ or
$\scenario_{\an_1}(\{a\}\{c\}\{b\})$.
\EOD
\end{example}

Scenarios induced by well-formed step sequences
are also scenarios in the structural sense.
Moreover,  maximal step sequences induce
maximal scenarios.

\begin{proposition}
\label{pr-icicied}
    Let $\sigma$  be a well-formed step sequence
    of an a\-cyc\-lic net $\an$.
\begin{enumerate}
\item
    $\scenario_\an(\sigma)\in\scenarios(\an)$
    and
    $\sigma\in\maxsseq(\scenario_\an(\sigma))$.
\item
    $\sigma\in\maxsseq(\an)$
    implies $\scenario_\an(\sigma)\in\maxscenarios(\an)$.
\end{enumerate}
\end{proposition}
\begin{proof}
    Let $\sigma$, $P$, and $T$ be as in
    Definition~\ref{def:13}.
    Moreover, let
    $\on=\scenario_\an(\sigma)$, and
    $\mu=M_0U_1M_1\dots M_{k-1}U_kM_k\in\mixsseq(\an)$
    be such that $\sigma=U_1\dots U_k$.

(1)
    By Definitions~\ref{def:6}, \ref{def:7rtt}, \ref{def:7aee}, and~\ref{def:13},
    we have
    $P=M_0\cup\dots\cup M_k$
    and $\PRE{\an}{T}\subseteq P$.
    Hence, by Definitions~\ref{def:4adddaa} and~\ref{def:4aaa},  $\on$ is an a\-cyc\-lic net
    such that $\on\sqsubseteq\an$.
    Moreover, by Definitions~\ref{def:10aa} and~\ref{def:5}(1),
    and Proposition~\ref{pr-jeff}(2), $\on$
    is an occurrence net.
    Hence $\on\in\scenarios(\an)$.
    It then follows from Proposition~\ref{def:5}(1),
    Definition~\ref{def:7aee}(1), and $\on\in\ON$
    that Proposition~\ref{pr-ueeu}(1,5) can be applied to conclude that
    $\sigma\in\maxsseq(\scenario_\an(\sigma))$.

(2)
    By part (1), $\on\in\scenarios(\an)$.
    If $\on\notin\maxscenarios(\an)$, then
    there is $\on'\in\ON$ such that
    $\on\sqsubseteq\on'\sqsubseteq\an$ and $T=T_\on\subset T_{\on'}$.
    By Proposition~\ref{pr-uususs},
    $\sigma\in\sseq(\on')$.
    Hence, by
    Proposition~\ref{pr-ueeu}(5), $\sigma\notin\maxsseq(\on')$.
    This produces a contradiction with
    $\sigma\in\maxsseq(\an)$ and Proposition~\ref{pr-uususs}.
\qed
\end{proof}

Following an observation made
in Example~\ref{ex-8}, two well-formed step sequences
induce the same scenario iff
they involve exactly the same transitions.

\begin{proposition}
\label{pr-icici}
    Let $\sigma$ and $\sigma'$ be well-formed step sequences
    of an a\-cyc\-lic net $\an$.
    Then
    $\scenario_\an(\sigma)=\scenario_\an(\sigma')$
    iff
    $\bigcup\sigma=\bigcup\sigma'$.
\end{proposition}
\begin{proof}
    It follows directly from Definition~\ref{def:13}.
\qed
\end{proof}

Well-formedness of a\-cyc\-lic nets can be re-phrased
in terms of their scenarios.

\begin{proposition}
\label{prop:10sd}
    The following statements are equivalent,
    for every a\-cyc\-lic net.
\begin{enumerate}
\item
    The a\-cyc\-lic net is well-formed.

\item
    Each transition occurs in at least one
    scenario, and each step sequence is a
    step sequence of at least one scenario.

\item
    Each transition occurs in at least one
    scenario, and each firing sequence is a
    firing sequence of at least one scenario.
\end{enumerate}
\end{proposition}
\begin{proof}
    Let $\an$ be an a\-cyc\-lic net.

$(1) \Longrightarrow (2)$
    Suppose that $\sigma\in\sseq(\an)$.
    Since $\sigma$ is well-formed, by Proposition~\ref{pr-icicied}(1),
    $\scenario_\an(\sigma)\in\scenarios(\an)$
    and
    $\sigma\in\sseq(\scenario_\an(\sigma))$.

    Suppose now that $t\in T_\an$.
    Since $\an$ is well-formed, there is
    $\sigma\in\sseq(\an)$ in which $t$ occurs.
    Moreover, as just shown, $\sigma\in\sseq(\scenario_\an(\sigma))$.

$(2) \Longrightarrow (3)$
    Obvious.

$(3) \Longrightarrow (1)$
    Suppose that not all step sequences of $\an$ are well-formed.
    Then there is $\sigma U\in\sseq(\an)$
    such that $\sigma$ is well-formed and $\sigma U$ is not.
    Hence, by Propositions~\ref{pr-hhhd}(2) and~\ref{pr-jeff}(3,4),
    there is a firing sequence
    $\tau\tau'\in\fseq(\an)$ such that
    $\bigcup \sigma=\bigcup\tau$ and $U=\bigcup\tau'$.
    Clearly, $\tau\tau'$ is not well-formed.
    On the other hand,
    there is $\on\in\scenarios(\an)$ such that
    $\tau\tau'\in\sseq(\on)$,
    contradicting   Proposition~\ref{prop:3xx}.
\qed
\end{proof}

Each (maximal) step sequence of a well-formed
a\-cyc\-lic net induces a  (maximal)
scenario
and, conversely, each (maximal) scenario  is induced by a
(maximal)
step sequence.
Therefore, the structurally and be\-ha\-viourally defined scenarios coincide.

\begin{proposition}
\label{prop:4tre}
    Let $\an$ be a well-formed a\-cyc\-lic net.
\begin{enumerate}
\item
    $\scenarios(\an)=\scenario_\an(\sseq(\an))$.
\item
    $\maxscenarios(\an)=\scenario_\an(\maxsseq(\an))$.
\end{enumerate}
\end{proposition}

\begin{proof}
    The ($\supseteq$) inclusions follow from Proposition~\ref{pr-icicied}.
    
(1)
    To show the ($\subseteq$) inclusion,
    suppose that $\on\in\scenarios(\an)$, and take any $\sigma\in\maxsseq(\on)$.
    By Propositions~\ref{pr-ueeu}(5) and~\ref{prop:8}, 
    $\bigcup\sigma=T_\on$ and $\sigma\in\sseq(\an)$.
    Hence $\scenario_\on(\sigma)=\on$ and $\scenario_\on(\sigma)=\scenario_\an(\sigma)$.
    As a result, $\on\in\scenario_\an(\sseq(\an))$.
     
(2)
    To show the ($\subseteq$) inclusion,
    suppose that $\on\in\maxscenarios(\an)$, and take any 
    $\sigma\in\maxsseq(\on)\subseteq\sseq(\an)$.
    By the proof of part (1), we have $\on\in\scenario_\an(\sseq(\an))$.
    Suppose that $\sigma\notin\maxsseq(\an)$. Then 
    there is $t\in T_\an$ such that $\sigma t\in\sseq(\an)$.
    We then observe that $\on\sqsubseteq\scenario_\an(\sigma t) $ and 
    $\on\neq\scenario_\an(\sigma t) $,
    contradicting the maximality of $\on$.
\qed
\end{proof}

Behaviours of   a\-cyc\-lic nets correspond  to the joint behaviour of
their scenarios.

\begin{proposition}
\label{prop:4trefff}
    Let $\an$ be a well-formed a\-cyc\-lic net.
\begin{enumerate}
\item
    $f(\an) =\bigcup f(\scenarios(\an))$,  \hfill
    for 
    $f=\sseq,\mixsseq,\reach,\fseq$.
\item
    $f(\an) =\bigcup f(\maxscenarios(\an))$, 
    \\\hspace*{\fill} 
    for  
    $f=\maxsseq,\maxmixsseq,\finreach$.
\end{enumerate}
\end{proposition}

\begin{proof}
(1)
    The ($\supseteq$) inclusions follow from Proposition~\ref{prop:8}.
    The  ($\subseteq$) inclusion for $f=\sseq$  
    follows from the second part of Proposition~\ref{pr-icicied}(1).
    The remaining   inclusions follow from this and
    Definition~\ref{def:4aaa}.

(2)
    It suffices to consider
    $f=\maxsseq$.
    
    To show the ($\subseteq$) inclusions, 
    suppose that  $\sigma\in\maxsseq(\an)$. Then, by Proposition~\ref{pr-icicied}(2),
    $\scenario_\an(\sigma)\in\maxscenarios(\an)$.
    Moreover, by Proposition~\ref{pr-icicied}(1),
    $\sigma\in\maxsseq(\scenario_\an(\sigma))$.
 
    To show the ($\supseteq$) inclusions,
    suppose that $\on\in \maxscenarios(\an)$ and 
    $\sigma\in\maxsseq(\on)\subseteq\sseq(\an)$.
    If $\sigma\notin\maxsseq(\an)$, then there is $t\in T_\an$
    such that $\sigma\in\sseq(\an)$.
    Hence $\on\sqsubseteq\scenario_\an(\sigma t) $ and 
    $\on\neq\scenario_\an(\sigma t) $,
    contradicting the maximality of $\on$. 
\qed
\end{proof}

In well-formed a\-cyc\-lic nets, all parts of their structure
are relevant as
they are covered by the scenarios in a graph-theoretic way.

\begin{proposition}
\label{pr:14eyu}
    $X_\an= \bigcup \{X_\on \mid \on\in\scenarios(\an)\}$,
    for every well-formed a\-cyc\-lic net $\an$ and  $X=P,T,F$.
\end{proposition}

\begin{proof}  
    The ($\supseteq$) inclusions follow from Definitions~\ref{def:4adddaa} and~\ref{def:4aaa}. 
    The ($\subseteq$) inclusions follow from Proposition~\ref{prop:10sd}
    and Definitions~\ref{def:4adddaa} and~\ref{def:4aaa}.
\qed
\end{proof}

\begin{proposition}
\label{prop:9fff} 
    $\scenarios(\an)\subseteq\scenarios(\an')$, for all a\-cyc\-lic nets 
    $\an$ and $\an'$ satisfying $\an\sqsubseteq\an'$.
\end{proposition}
\begin{proof}
    It follows directly from the definitions.
\qed
\end{proof} 

The next result shows that an occurrence net has exactly 
one maximal scenario (itself),
which can be interpreted as saying it is a precise representation 
of a single execution history.
\begin{proposition}
\label{prop:9}
    $\maxscenarios(\on)=\{\on\}$, for every occurrence $\on$.
\end{proposition}
\begin{proof}
    By Definition~\ref{def:11},
    $\on\in\maxscenarios(\on)$.

    Suppose that $\on'\in\maxscenarios(\on)\setminus\{\on\}$. Then,
    $\on'\sqsubseteq\on\sqsubseteq\on$ and $\on'\neq\on$.
    This, however, yields a contradiction with Definition~\ref{def:11}.
\qed
\end{proof}

\section{Communication structured a\-cyc\-lic nets }

Communication Structured Acyclic Nets (\textsc{csa}-nets) are 
a generalisation of Communication Structured Occurrence Nets (\textsc{cso}-nets) 
introduced and discussed in~\cite{c51,c20,c1}. 
In~\cite{c4,c2}, \textsc{cso}-nets were used to to deal with cybercrime and major accidents, and~\cite{c11} 
demonstrated \textsc{cso}-nets can be used as a framework for visualising and analysing behaviour of complex evolving systems. 
Other works on \textsc{cso}-nets were related to provenance~\cite{c23} and timed behaviours~\cite{c24}.
This section introduces nets which
generalise \textsc{cso}-nets
by being based on a\-cyc\-lic nets rather than occurrence nets.

The following are intuitive explanations of the main
concepts:
\begin{description}

\item[communication structured a\-cyc\-lic net (\textsc{csa}-net)]{\ }
\\
A `database' consisting of a number of disjoint
a\-cyc\-lic nets which communicate through
special buffer/channel places.
\textsc{csa}-nets can exhibit backward and forward non-de\-ter\-mi\-nism.
They can contain cycles restricted to the buffer places.
 
\item[buffer place]{\ }
\\
A place allowing synchronous or asynchronous transfer of tokens.
A cycle involving only buffer places
`implements' synchronous communication.
 
\item[communication structured occurrence net (\textsc{cso}-net)]{\ }
\\
A \textsc{csa}-net providing full and unambiguous  record of all causal
dependencies between the events involved in  a single `causal history'.

\item[backward de\-ter\-mi\-nistic \textsc{csa}-net (\textsc{bdcsa}-net)]{\ }
\\
A \textsc{csa}-net net providing unambiguous record of causal
dependencies.

\item[scenario]{\ }
\\
A \textsc{cso}-net  providing a representation of
a  system history.

\item[well-formed \textsc{csa}-net]{\ }
\\ 
A \textsc{csa}-net where each execution history 
from the initial state has an unambiguous interpretation in terms of 
causality and concurrency.

\item[projected mixed step sequence]{\ }
\\
A mixed step sequence restricted to the nodes of a component acyclic net.

\item[projected step sequence]{\ }
\\
A   step sequence restricted to the  
transitions of a component acyclic.
\end{description}

\begin{figure}[ht!]
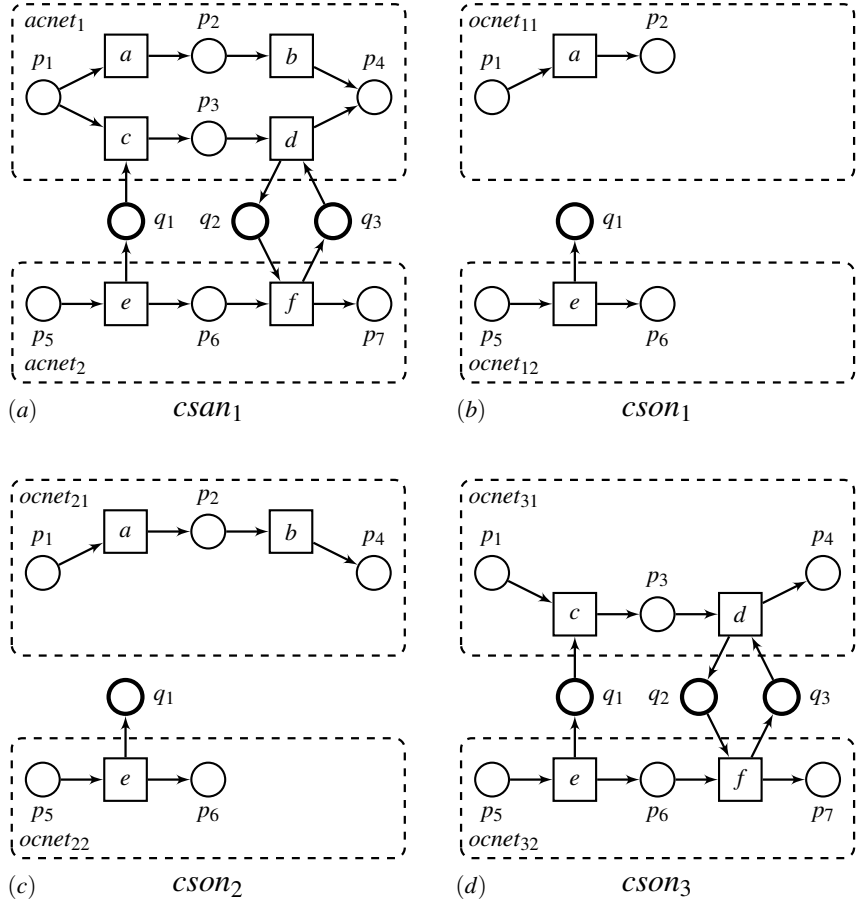

\begin{center}
\StandardNet[0.55]
\Put{-.5}{-2.6}{$(a)$}
\Put{4}{-2.6}{{\large $\csan_1$}}
\Put{0.3}{6.8}{$\an_1$}
\Put{0.3}{-1.5}{$\an_2$}

\draw [dashed,rounded corners]
   (-.75, 3) --
   (-.75, 7.25) --
   (8.75, 7.25) --
   (8.75, 3) -- cycle;
\draw [dashed,rounded corners]
   (-.75,-1.9) --
   (-.75, 1) --
   (8.75, 1) --
   (8.75,-1.9) -- cycle;

    \PlacN{P1}{ 0}{5}{}{p_1}
    \PlacN{P2}{ 4}{6}{}{p_2}
    \PlacN{P3}{ 4}{4}{}{p_3}
    \PlacN{P4}{ 8}{5}{}{p_4}

    \Whitetran{A}{ 2}{6}{a}
    \Whitetran{B}{ 6}{6}{b}
    \Whitetran{C}{ 2}{4}{c}
    \Whitetran{D}{ 6}{4}{d}

    \diredge{P1}{A}
    \diredge{P1}{C}
    \diredge{P2}{B}
    \diredge{P3}{D}

    \diredge{A}{P2}
    \diredge{B}{P4}
    \diredge{D}{P4}
    \diredge{C}{P3}

    \PlacS{R1}{ 0}{0}{}{p_5}
    \PlacS{R2}{ 4}{0}{}{p_6}
    \PlacS{R3}{ 8}{0}{}{p_7}

    \Whitetran{E}{ 2}{0}{e}
    \Whitetran{F}{ 6}{0}{f}

    \diredge{R1}{E}
    \diredge{R2}{F}
    \diredge{E}{R2}
    \diredge{F}{R3}

    \bPlacE{Q1}{ 2}{2}{}{q_1}
    \bPlacW{Q2}{ 5}{2}{}{q_2}
    \bPlacE{Q3}{ 7}{2}{}{q_3}

    \diredge{Q1}{C}
    \diredge{E}{Q1}
    \diredge{D}{Q2}
    \diredge{Q2}{F}
    \diredge{Q3}{D}
    \diredge{F}{Q3}
\end{tikzpicture}
~~~~
\StandardNet[0.55]
\Put{-.5}{-2.6}{$(b)$}
\Put{4}{-2.6}{{\large $\cson_1$}}
\Put{0.3}{6.8}{$\on_{11}$}
\Put{0.3}{-1.5}{$\on_{12}$}

\draw [dashed,rounded corners]
   (-.75, 3) --
   (-.75, 7.25) --
   (8.75, 7.25) --
   (8.75, 3) -- cycle;
\draw [dashed,rounded corners]
   (-.75,-1.9) --
   (-.75, 1) --
   (8.75, 1) --
   (8.75,-1.9) -- cycle;

    \PlacN{P1}{ 0}{5}{}{p_1}
    \PlacN{P2}{ 4}{6}{}{p_2}

    \Whitetran{A}{ 2}{6}{a}

    \diredge{P1}{A}

    \diredge{A}{P2}

    \PlacS{R1}{ 0}{0}{}{p_5}
    \PlacS{R2}{ 4}{0}{}{p_6}

    \Whitetran{E}{ 2}{0}{e}

    \diredge{R1}{E}
    \diredge{E}{R2}

    \bPlacE{Q1}{ 2}{2}{}{q_1}
    \diredge{E}{Q1}
\end{tikzpicture}
\\[6mm]
\StandardNet[0.55]
\Put{-.5}{-2.6}{$(c)$}
\Put{4}{-2.6}{{\large $\cson_2$}}
\Put{0.3}{6.8}{$\on_{21}$}
\Put{0.3}{-1.5}{$\on_{22}$}

\draw [dashed,rounded corners]
   (-.75, 3) --
   (-.75, 7.25) --
   (8.75, 7.25) --
   (8.75, 3) -- cycle;
\draw [dashed,rounded corners]
   (-.75,-1.9) --
   (-.75, 1) --
   (8.75, 1) --
   (8.75,-1.9) -- cycle;

    \PlacN{P1}{ 0}{5}{}{p_1}
    \PlacN{P2}{ 4}{6}{}{p_2}
    \PlacN{P4}{ 8}{5}{}{p_4}

    \Whitetran{A}{ 2}{6}{a}
    \Whitetran{B}{ 6}{6}{b}

    \diredge{P1}{A}
    \diredge{P2}{B}

    \diredge{A}{P2}
    \diredge{B}{P4}

    \PlacS{R1}{ 0}{0}{}{p_5}
    \PlacS{R2}{ 4}{0}{}{p_6}

    \Whitetran{E}{ 2}{0}{e}

    \diredge{R1}{E}
    \diredge{E}{R2}

    \bPlacE{Q1}{ 2}{2}{}{q_1}
    \diredge{E}{Q1}
\end{tikzpicture}
~~~~
\StandardNet[0.55]
\Put{-.5}{-2.6}{$(d)$}
\Put{4}{-2.6}{{\large $\cson_3$}}
\Put{0.3}{6.8}{$\on_{31}$}
\Put{0.3}{-1.5}{$\on_{32}$}

\draw [dashed,rounded corners]
   (-.75, 3) --
   (-.75, 7.25) --
   (8.75, 7.25) --
   (8.75, 3) -- cycle;
\draw [dashed,rounded corners]
   (-.75,-1.9) --
   (-.75, 1) --
   (8.75, 1) --
   (8.75,-1.9) -- cycle;

    \PlacN{P1}{ 0}{5}{}{p_1}
    \PlacN{P3}{ 4}{4}{}{p_3}
    \PlacN{P4}{ 8}{5}{}{p_4}

    \Whitetran{C}{ 2}{4}{c}
    \Whitetran{D}{ 6}{4}{d}

    \diredge{P1}{C}
    \diredge{P3}{D}
    \diredge{D}{P4}
    \diredge{C}{P3}

    \PlacS{R1}{ 0}{0}{}{p_5}
    \PlacS{R2}{ 4}{0}{}{p_6}
    \PlacS{R3}{ 8}{0}{}{p_7}

    \Whitetran{E}{ 2}{0}{e}
    \Whitetran{F}{ 6}{0}{f}

    \diredge{R1}{E}
    \diredge{R2}{F}
    \diredge{E}{R2}
    \diredge{F}{R3}

    \bPlacE{Q1}{ 2}{2}{}{q_1}
    \bPlacW{Q2}{ 5}{2}{}{q_2}
    \bPlacE{Q3}{ 7}{2}{}{q_3}

    \diredge{Q1}{C}
    \diredge{E}{Q1}
    \diredge{D}{Q2}
    \diredge{Q2}{F}
    \diredge{Q3}{D}
    \diredge{F}{Q3}
\end{tikzpicture}

\end{center}
\caption{Communication structured a\-cyc\-lic nets.}
\label{fig-2x}
\end{figure}

\begin{definition}[\textsc{csa}-net]
\label{def:20}
    A \emph{communication structured a\-cyc\-lic net (or \textsc{csa}-net)}
    is a tuple
    $\csan= (\an_1,\dots,\an_n,Q,W)$ ($n\geq 1$)
    such that:
\begin{enumerate}
\item
    $\an_1,\dots,\an_n$ are well-formed a\-cyc\-lic nets
    with disjoint sets of nodes (i.e., places and transitions).
    We also denote:
\[
\begin{array}{lcl@{~~~~~~~~~}lcl}
    P_\csan
    & \triangleq
    & P_{\an_1}\cup\dots\cup P_{\an_n}
    & Q_\csan
    & \triangleq
    & Q
\\
    T_\csan
    & \triangleq
    & T_{\an_1}\cup\dots\cup T_{\an_n}
    & W_\csan
    & \triangleq
    & W
\\
    F_\csan
    & \triangleq
    & F_{\an_1}\cup\dots\cup F_{\an_n}
    & \net_{\csan,i}
    & \triangleq
    & \an_i ~~~(\textit{for}~~1\leq i\leq n)
\;.
\end{array}
\]
\item
    $Q$ is a finite set of \emph{buffer places} disjoint from
    $P_\csan\cup T_\csan$.
    
\item
    $W\subseteq (Q\times T_\csan)\cup (T_\csan\times Q)$
    is a set of arcs.
\item
    For every buffer place $q$:
\begin{enumerate}
\item
    There is at least one transition $t$ such that $t W q$.
\item
    If $t W q$ and $q W u$
    then transitions $t$ and $u$
    belong  to different component a\-cyc\-lic nets.
\end{enumerate}
\end{enumerate}
    Notation: $\CSAN$ is the set of all \textsc{csa}-nets.
\EOD
\end{definition}
That is, in addition to requiring the disjointness of the component
acyclic nets and the buffer places, it is required that
buffer places pass tokens between different a\-cyc\-lic nets.
Hence, if $n=1$, then $\csan=(\an_1,\es,\es)$ and so the 
\textsc{csa}-net $\csan$  can be
identified with its only component acyclic net $\an_1$
(not only structurally, but also syntactically, as will be 
shown later). 

\begin{notation}[direct precedence in \textsc{csa}-net]
\label{not:1a}
    Let $\csan= (\an_1,\dots,\an_n,Q,W)$
    be a \textsc{csa}-net,
    $x\in P_\csan\cup T_\csan\cup Q_\csan$, and
    $X\subseteq P_\csan\cup T_\csan\cup Q_\csan$. Then
\[
\begin{array}{lll@{~~~~~~~~}lll}
    \PRE{\csan}{x}
    & \triangleq
    & \{y
      \mid y\,F_\csan\,x\;\vee\; y\,W_\csan\,x\}
    & \PRE{\csan}{X}
    & \triangleq
    & \bigcup\{\PRE{\csan}{z}\mid z\in X\} 
\\
    \POST{\csan}{x}
    & \triangleq
    & \{y
      \mid x\,F_\csan\,y\;\vee\; x\,W_\csan\,y\}
    & \POST{\csan}{X}
    & \triangleq
    & \bigcup\{\POST{\csan}{z}\mid z\in X\} 
\end{array}
\]
    denote direct predecessors and successors of $x$ and $X$, respectively.
    Moreover, 
\[
    P_\csan^\init\triangleq P_{\an_1}^\init\cup\dots\cup P_{\an_n}^\init
    ~~\mbox{and}~~
    P_\csan^\fin\triangleq P_{\an_1}^\fin\cup\dots\cup P_{\an_n}^\fin 
    \cup (Q_\csan\setminus\PRE{\csan}{T_\csan})
\]    
    are the \emph{initial} and \emph{final} places of $\csan$, respectively.
\EOD
\end{notation}
Note that no buffer place is an initial place. 

\begin{proposition}
\label{prop-vxxx}
    $P_\csan\uplus Q_\csan 
    =P_\csan^\init\uplus\POST{\csan}{T_\csan}
    =P_\csan^\fin\uplus\PRE{\csan}{T_\csan}$,
   for every \textsc{csa}-net $\csan$.
\end{proposition}
\begin{proof}
    It follows directly from the definitions.
    \qed
\end{proof}

\begin{example}
\label{ex-1ede}
Figure~\ref{fig-2x}($a$) depicts a \textsc{csa}-net $\csan_1$ such that
$\PRE{\csan_1}{p_4}=\{b,d\}$ and
$\POST{\csan_1}{e}=\{q_1,p_6\}$.
Moreover, $P^\init_{\csan_1}=\{p_1,p_5\}$.
\EOD
\end{example}

\subsection{Subnets in \textsc{csa}-nets}

The next notion captures a structural inclusion
between \textsc{csa}-nets. 

\begin{definition}[subnet of \textsc{csa}-net]
\label{def:4axxyd}
    Let $\csan= (\an_1,\dots,\an_n,Q,W)$ and 
    $\csan'= (\an'_1,\dots,\an'_n,Q',W')$ be \textsc{csa}-nets.
    Then $\csan$ is \emph{included} in $\csan'$
    if the following hold:
\begin{enumerate}
\item
    $\an_i\subseteq\an'_i$, for every $1\leq i\leq n$. 
\item
    $Q\subseteq Q'$.
\item
    $\PRE{\csan}{t}=\PRE{\csan'}{t}$ and
    $\POST{\csan}{t}=\POST{\csan'}{t}$,
    for every $t\in T_\csan$.
\end{enumerate}
    Notation: $\csan\subseteq\csan'$.
\EOD
\end{definition}
Note that a transition included in a subnet 
retains its pre-places and post-places, including the buffer places; in other words, it retains 
its local environment. 
 
\begin{proposition}
\label{pr-uusugfdddee}
    Let $\csan\subseteq\csan'$ be \textsc{csa}-nets.
    Then
    $Q_\csan\subseteq Q_{\csan'}$  and $W_\csan\subseteq W_{\csan'}$.
\end{proposition}
\begin{proof}
    It suffices to show 
    $Q_\csan\subseteq Q_{\csan'}$. Let $q\in Q_\csan$. By
    Definition~\ref{def:20}(4a), there is $t\in T_\csan$ such that 
    $t W_\csan q$. Hence, by Definition~\ref{def:4axxyd}(2), $t W_{\csan'} q$.
    Thus $q\in Q_{\csan'}$.
\qed
\end{proof}

The next notion captures not only a structural inclusion
between \textsc{csa}-nets, but it is
also intended to correspond to inclusion
of the behaviours they capture.

\begin{definition}[co-initial subnet of \textsc{csa}-net]
\label{def:4aafffa} 
    A \textsc{csa}-net 
    $\csan$ is a \emph{co-initial subnet} of a \textsc{csa}-net $\csan'$
    if $\csan\subseteq\csan'$ and 
    $P_\csan^\init=P_{\csan'}^\init$.
    \\
    Notation: $\csan\sqsubseteq\csan'$.
\EOD
\end{definition}

Note that if $\csan$ and $\csan'$ are as in Definition~\ref{def:4axxyd}, then 
$\csan\sqsubseteq\csan'$ iff $\an_i\sqsubseteq\an'_i$, for every $1\leq i\leq n$.
 
\begin{proposition}
\label{pr-uusugfddd}
    Let $\csan\sqsubseteq\csan'$ be \textsc{csa}-nets.
\begin{enumerate}  
\item
    $P_\csan=P_{\csan'}^\init\uplus \POST{\csan'}{T_\csan}$. 
\item
    $\csan=\csan'$ iff $T_\csan=T_{\csan'}$.
\end{enumerate}
\end{proposition}
\begin{proof}
(1) It follows from Proposition~\ref{pr-uusu}(1) as well as
    Definitions~\ref{def:20}(4a) and~\ref{def:4axxyd}(2).
    
(2) 
    The ($\Longrightarrow$) implication is obvious,
    and the ($\Longleftarrow$) implication follows from  Proposition~\ref{pr-uusu}(2)
    as well as
    Definitions~\ref{def:20}(4a), \ref{def:4axxyd}(2), and~\ref{def:4aafffa}.
\qed
\end{proof}

\subsection{Subclasses of \textsc{csa}-nets}

The next definition introduces \textsc{csa}-nets 
representing individual causal histories.

\begin{definition}[\textsc{cso}-net]
\label{def:22}
    A \emph{communication structured occurrence net (or \textsc{cso}-net)}
    is $\cson\in\CSAN$
    such that:
\begin{enumerate}
\item
    The component a\-cyc\-lic nets belong to $\ON$.  
\item
    $|\PRE{\cson}{q}|=1$ and $|\POST{\cson}{q}|\leq 1$,
    for every $q\in Q_\cson$.
\item
    No place in $P_\cson$ belongs to a cycle in the graph of $F_\cson\cup W_\cson$. 
\end{enumerate}
    Notation: $\CSON$ is the set of all \textsc{cso}-nets.
\EOD
\end{definition}
That is, only cycles involving buffer places are allowed.
\textsc{cso}-nets  exhibit backward de\-ter\-mi\-nism and forward
de\-ter\-mi\-nism.

\begin{example}
\label{ex-1ee}
Figure~\ref{fig-2x}($b,c,d$) depict  \textsc{cso}-nets.
\EOD
\end{example}

We also consider \textsc{csa}-nets with only forward nondeterminism.

\begin{definition}[\textsc{bdcsa}-net]
\label{def:21}
    A \emph{backward de\-ter\-mi\-nistic communication structured a\-cyc\-lic net
    (or \textsc{bdcsa}-net)}
    is
    $\bdcsan\in\CSAN$
    such that:
\begin{enumerate}
\item
    The component a\-cyc\-lic nets belong to $\BDAN$.  
\item
    $|\PRE{\bdcsan}{q}|=1$,
    for every $q\in Q_\bdcsan$. 
\end{enumerate}
    Notation: $\BDCSAN$ is the set of all \textsc{bdcsa}-nets.
\EOD
\end{definition}
\textsc{bdcsa}-nets  exhibit backward de\-ter\-mi\-nism, but forward
de\-ter\-mi\-nism is not required.

\begin{proposition}
\label{prop:2x}
    $\CSON\subset\BDCSAN\subset\CSAN$.
\end{proposition}
\begin{proof}
    It follows directly from Proposition~\ref{prop:2} and the definitions.
\qed
\end{proof}

\section{Step sequence semantics of \textsc{csa}-nets}

The intuition behind the step sequence semantics of \textsc{csa}-net
is similar as in the case of acyclic nets, and so the terminology used
will also be similar.
 
\begin{definition}[step and marking of \textsc{csa}-net]
\label{def:5x}
    Let $\csan$ be a \textsc{csa}-net.
\begin{enumerate}
\item
    $\STEPS(\csan)\triangleq \{U\in \pset{T_\csan}\setminus\{\es\}
      \mid
      \forall t\neq u\in U: \PRE{\csan}{t}\cap\PRE{\csan}{u}=\es\}$
    are the \emph{steps}.
\item
    $\MARK(\csan)\triangleq \pset{P_\csan\cup Q_\csan}$ are the \emph{markings}.
\item
    $M_\csan^\init\triangleq P_\csan^\init$
    is the \emph{initial} marking. 
\EOD
\end{enumerate}
\end{definition}

\begin{example}
\label{ex-2sx}
$\STEPS(\csan_1)=\{U\in \pset{\{a,b,c,d,e,f\}}
\setminus\{\es\}\mid a\notin U\vee c\notin U\}$, for
the  \textsc{csa}-net in Figure~\ref{fig-2x}($a$).
\EOD
\end{example}

\begin{definition}[enabled and executed step of \textsc{csa}-net]
\label{def:6x}
    Let $M$ be a marking of a \textsc{csa}-net $\csan$.
\begin{enumerate}
\item
    $\enabled_\csan(M)\triangleq
    \{U\in\STEPS(\csan)\mid \PRE{\csan}{U}\subseteq M\cup(\POST{\csan}{U}\cap Q)\}$
    are the steps \emph{enabled} at $M$.
\item
    A step $U\in\enabled_\csan(M)$
    can be \emph{executed} yielding a new
    marking
\[
    M'\triangleq (M\cup\POST{\csan}{U})\setminus\PRE{\csan}{U}\;.
\]
    Notation:  $M\STEP{U}{\csan}M'$.
\EOD
\end{enumerate}
\end{definition}
Enabling a step in a \textsc{csa}-net amounts to having all input places
belonging to the component a\-cyc\-lic nets present/marked in a
global state.
Moreover, if an input buffer place is not present,
then it must be an output place of a transition belonging to the step.
Such a mechanism allows one to sychronise transitions coming from different
component a\-cyc\-lic nets. The same mechanism of simultaneous output to and input
from a place is not available within the component a\-cyc\-lic nets, and
so in Definition~\ref{def:6x}(1) we have
$\PRE{\csan}{U}\subseteq M\cup(\POST{\csan}{U}\cap Q)$ rather than
$\PRE{\csan}{U}\subseteq M\cup \POST{\csan}{U} $.

\begin{definition}[mixed step sequence and step sequence of \textsc{csa}-net]
\label{def:7a}
    Let $M_0,\dots,M_k$ ($k\geq 0$) be
    markings and $U_1,\dots,U_k$ be steps
    of a \textsc{csa}-net $\csan$ such that
    $M_{i-1}\STEP{U_i}{\csan}M_i$,
    for every $1\leq i\leq k$.
\begin{enumerate}
\item
    $\mu=M_0U_1M_1\dots M_{k-1}U_k M_k$
    is a \emph{mixed step sequence from $M_0$ to $M_k$}.
\item
    $\sigma=U_1\dots U_k$ is
    a \emph{step sequence from $M_0$ to $M_k$}.
\end{enumerate}
    The above two notions are denoted by $M_0\mixSTEP{\mu}{\csan}M_k$
    and $M_0\STEP{\sigma}{\csan}M_k$, respectively.
    Moreover,
    $M_0\STEP{\sigma}{\csan}$ denotes
    that $\sigma$ is a step sequence \emph{enabled $M_0$}, and
    $M_0\STEP{}{\csan}M_k$
    denotes that $M_k$ is \emph{reachable from $M_0$}.
\EOD  
\end{definition}
If $k=0$ then $\mu=M_0$ and the corresponding
step sequence $\sigma$ is the \emph{empty} sequence denoted by $\lambda$.
 
In the last definition, the starting point
is an arbitrary marking.
The next definition assumes that the starting point is
the default initial marking.

\begin{definition}[behaviour of \textsc{csa}-net]
\label{def:7}
    The following sets capture various be\-ha\-vioural notions
    related to step sequences and reachable markings
    of a \textsc{csa}-net $\csan$.
\begin{enumerate}
\item
    $\sseq(\csan)\triangleq\{\sigma\mid  M_\csan^\init\STEP{\sigma}{\csan}M\}$
    \hfill \emph{step sequences}.

\item
    $\mixsseq(\csan)\triangleq\{\mu\mid  M_\csan^\init\mixSTEP{\mu}{\csan}M\}$
    \hfill   \emph{mixed step sequences}.

\item
    $\maxsseq(\csan)\triangleq\{\sigma\in\sseq(\csan)\mid
         \neg\exists U: \sigma U\in\sseq(\csan)\}$
    \\\hspace*{\fill}   \emph{maximal step sequences}.

\item
    $\maxmixsseq(\csan)\triangleq\{\mu\in\mixsseq(\csan)\mid
         \neg\exists U,M: \mu UM\in\mixsseq(\csan)\}$
    \\\hspace*{\fill} \emph{maximal mixed step sequences}.

\item
    $\reach(\csan)\triangleq \{M\mid
               M_\csan^\init\STEP{ }{\csan}M\}$
    \hfill   \emph{reachable markings}.

\item
    $\finreach(\csan)\triangleq\{M\mid \exists \sigma\in\maxsseq(\csan):
               M_\csan^\init\STEP{\sigma}{\csan}M\}$
    \\\hspace*{\fill}  \emph{final reachable markings}.
\end{enumerate}
\EOD
\end{definition}

Note that if a \textsc{csa}-net has only one component acyclic net, $\an$, then 
$f(\csan) = f(\an)$, for $f=\sseq,\mixsseq,\maxsseq,\reach,\finreach$.

\begin{example}
\label{ex-5jw}
    The following hold for the \textsc{csa}-net in Figure~\ref{fig-2x}($a$).
\begin{enumerate}
\item
    $\sseq(\csan_1)=\{\lambda,\{a,e\} ,\{a,e\}b, \cdots\}$.

\item
    $\mixsseq(\csan_1)=\{
    \{p_1,p_5\},
    \{p_1,p_5\}\{a,e\}\{p_2,p_6,q_1\},
    \dots\}$.

\item
    $\maxsseq(\csan_1)=\{
     ab,aeb,abe,\{a,e\}b,a\{b,e\},
     ec\{d,f\},\{e,c\}\{d,f\}\}$.

\item
    $\maxmixsseq(\csan_1)=\{
    \{p_1,p_5\}\{a,e\}\{p_2,p_6,q_1\}b\{p_4,p_6,q_1\} ,\dots\}$.

\item
    $\reach(\csan_1)=\{\{p_1,p_5\}, \{p_2,p_6,q_1\},\dots\}$.

\item
    $\finreach(\csan_1)=\{
    \{p_4,p_6,q_1\},\{p_4,p_7\}\}$.
\EOD
\end{enumerate}
\end{example}

Projecting step sequences and mixed step sequences of
a \textsc{csa}-net onto
a component a\-cyc\-lic net yields
sequences,   mixed step sequences and scenarios of the latter.

\begin{definition}[projected (mixed) step sequence of \textsc{csa}-net]
\label{de-77}
    Let $\sigma$ be a step sequence of a \textsc{csa}-net $\csan$,
    $\mu$ be a mixed step sequence of $\csan$, and $\an$
    be one of the component a\-cyc\-lic nets of $\csan$.
\begin{enumerate}
\item
    $\proj_\an(\sigma)=\sigma\upharpoonright_{T_\an}$.
\item
    $\proj_\an(\mu)$ is obtained from
    $\mu\upharpoonright_{P_\an\cup T_\an}$ by deleting consecutive duplicate sets.
\end{enumerate}
\end{definition}

\begin{example}
\label{ex-33}
For $\csan_1$ in Figure~\ref{fig-2x}($a$),
$\proj_{\an_1}(\{a,e\}b)=ab$, $\proj_{\an_2}(\{a,e\}b)=e$, and
$\proj_{\an_2}(\{p_1,p_5\}a\{p_2,p_5\}b\{p_4,p_5\})=\{p_5\}$.
\EOD
\end{example}

Projected step sequences are step sequences of the component
a\-cyc\-lic nets.

\begin{proposition}
\label{pr-883}
    For
    every component a\-cyc\-lic net $\an$ of 
    a \textsc{csa}-net $\csan$:
\[
\begin{array}{lcl}
    \proj_\an(\sseq(\csan))
    & \subseteq
    & \sseq(\an)
\\
    \proj_\an(\mixsseq(\csan))
    & \subseteq
    & \mixsseq(\an)\;.
\end{array}
\]
\end{proposition}
\begin{proof}
    It suffices to demonstrate
    that we have $\proj_\an(\mu)\in\mixsseq(\an)$,
    where $\mu=M_0U_1M_1\dots M_{k-1}U_k M_k$.
    We observe that in such a case the enabledness and calculations of 
    the results of step executions
    in $\csan$ and $\an$ are fully consistent except that if 
    $U_i\cap T_\an=\es$, then $M_{i-1}\cap P_\an=M_i\cap P_\an$
    (since $\an$ is well-formed, this is
    the only situation that $M_{i-1}\cap P_\an=M_i\cap P_\an$). 
    But then the construction of $\proj_\an(\mu)$ then deletes multiple neighbouring copies 
    of the same marking.
\qed
\end{proof}
 
\subsection{Behaviour of subnets of \textsc{csa}-nets}

The structural inclusion of  
\textsc{csa}-nets implies the inclusion of the behaviours
they generate.

\begin{proposition}
\label{pr-uusussddd}
    Let $\csan\sqsubseteq\csan'$ be \textsc{csa}-nets.
    Then we have $f(\csan) \subseteq f(\csan')$, for  
    $f=\mixsseq,\sseq, \reach$.
\end{proposition}
\begin{proof}
    Similar to the proof of Proposition~\ref{pr-uususs}.
\qed
\end{proof} 

Due to being both backward-deterministic and
forward-deterministic, \textsc{cso}-nets enjoy 
several useful be\-ha\-vioural properties similar 
to those satisfied by occurrence nets.
For example, each
\textsc{cso}-net has a step sequence 
which uses all  
the transitions.
This and other properties are gathered together 
in the next result.

\begin{proposition}
\label{pr-ueeudd}
    Let $\sigma=U_1\dots U_k$ ($k\geq 0$) be a
    sequence of nonempty sets of transitions of \textsc{cso}-net $\cson$.
    Moreover, let $M$ be a reachable marking of $\cson$ and
    $T=\bigcup\sigma$.
\begin{enumerate}
\item
    $\sigma\in\sseq(\cson)$ iff for every $1\leq i\leq k$,
    $U_i\cap (U_1\cup\dots \cup U_{i-1})=\es$
    and 
\[
    \PRE{\cson}{\PRE{\cson}{U_i}}\subseteq U_1\cup\dots\cup U_{i-1}\cup 
    U_i\cap \PRE{\cson}{Q\cap\PRE{\cson}{U_i}}\;.
\]
\item
    $\sigma\in\sseq(\cson)$
    implies
    $M_\cson^\init\STEP{\sigma}{\cson}
      (M_\cson^\init\cup \POST{\cson}{T})
        \setminus \PRE{\cson}{T}$.
\item
    $T_\cson=\bigcup\{\bigcup\xi\mid \xi\in\sseq(\cson)\}$.
\item
    $M\STEP{U}{\cson}M'$ implies $U'\in\enabled_\cson(M')$, for every
    $U'\in\enabled_\cson(M)$ such that $U\cap U'=\es$. 
\item
    $\maxsseq(\cson)=\{\xi\in\sseq(\cson)\mid \bigcup\xi=T_\cson\}$.
\item
    $\finreach(\cson)=\{P^\fin_\cson\}$.
\end{enumerate}
\end{proposition}
\begin{proof}
     If follows from the results proved in~\cite{c1}.
\qed
\end{proof}
 
In general, the step sequences of \textsc{cso}-nets cannot be fully 
sequentialised (or linearised) due to the presence of synchronous 
communications which cannot be split. The next 
definition captures this notion formally.

\begin{definition}[syn-cycle]
\label{def-syncycle}
    A \emph{synchronous cycle (or syn-cycle)} of a \textsc{cso}-net
    $\cson$ is a maximal set of transitions $S\subseteq T_\cson$
    such that, for all $t\neq u\in S$, $(t,u)\in W^+_\cson$.
    \\
    The set of all sync-cycles is denoted by $\syncycles(\cson)$.
\EOD
\end{definition}

The idea behind the notion of syn-cycles is to identify  
`tight' synchronous communications which cannot be executed 
in stages.  

\begin{example}
\label{ex-33dd}
For $\cson_3$ in Figure~\ref{fig-2x}($d$),
$\syncycles(\cson_3)=\{\{c\},\{e\},\{d,f\}\}$.
\EOD
\end{example}

Syn-cycles are nonempty and form a partition of the set of transitions.

\begin{proposition}
\label{prop-uhegr}
     $\syncycles(\cson)$ forms a partition of $T_\cson$, for every
     \textsc{cso}-net $\cson$.
\end{proposition}
\begin{proof}
     If follows from results proved in~\cite{c1}.
\qed
\end{proof}

Each step occurring in a step sequence of a \textsc{cso}-net 
can be partitioned into syn-cycles (in a unique way). 

\begin{proposition}
\label{prop-uhddegr}
     Let $M$ be a reachable marking of a  \textsc{cso}-net $\cson$
     and $M\STEP{U}{\cson}M'$.
     Then there are syn-cycles
     $S_1,\dots, S_k\in\syncycles(\cson)$ 
     such that $U=S_1\uplus\dots\uplus S_k$ and 
     $M\STEP{S_1\dots S_k}{\cson}M'$.
\end{proposition}
\begin{proof}
     If follows from results proved in~\cite{c1}.
\qed
\end{proof}

This means, for example, that all reachable markings of a \textsc{cso}-net
can be generated by executing syn-cycles rather than all potential steps.

\section{Well-formed \textsc{csa}-nets}

A basic consistency criterion applied to \textsc{csa}-nets
is well-formedness. Its essence is to
ensure a clean representation of causality in
behaviours they represent.
The definition of a well-formed \textsc{csa}-net is derived from the notion of a well-formed step
sequence.

\begin{definition}[well-formed step sequence of \textsc{csa}-net]
\label{def:10aadd}
    A step sequence $U_1\dots U_k$ of
    a \textsc{csa}-net $\csan$ is \emph{well-formed}
    if the following hold:
\begin{enumerate}
\item
    $\POST{\csan}{t}\cap\POST{\csan}{u}=\es$,
    for every $1\leq i\leq k$ and all $t\neq u\in U_i$.
\item
    $\POST{\csan}{U_i}\cap\POST{\csan}{U_j}=\es$,
    for all $1\leq i<j\leq k$.
\EOD
\end{enumerate}
\end{definition}
Intuitively, in a well-formed step sequence of a \textsc{csa}-net, 
no place is filled by a token  more than once.

It then follows, e.g., that in
a well-formed step sequence, no token is consumed more than once,
no transition is executed more than once,
the order of execution of transitions does not influence the
resulting marking.
 
\begin{proposition}
\label{pr-jeffdd}
    Let $\sigma=U_1\dots U_k$ be a well-formed
    step sequence of a  \textsc{csa}-net $\csan$,
    and
    $M_\csan^\init=M_0 U_1 M_1 \dots M_{k-1} U_k M_k$ be the corresponding 
    mixed step sequence.
    Moreover, let $T=\bigcup\sigma$.
\begin{enumerate}

\item
    $\PRE{\csan}{U_i}\cap\POST{\csan}{U_i}\subseteq Q_\csan$,
    for every $1\leq i\leq k$.

\item
    $\PRE{\csan}{U_i}\cap\PRE{\csan}{U_j}=\es$ and $U_i\cap U_j=\es$,
    for all $1\leq i<j\leq k$.

\item
    $M_k=M_\csan^\init\cup \POST{\csan}{T}\setminus \PRE{\csan}{T}$. 
\end{enumerate}
\end{proposition}

\begin{proof} 
    Similar to the proof of Proposition~\ref{pr-jeff}(1,2,3).
\qed
\end{proof}

Well-formedness is ensured when backward non-determinism is
not allowed. 

To develop a sound treatment of
causality in \textsc{csa}-nets, it is  required that
all step sequences are well-formed.
Moreover, it is required that a well-formed a\-cyc\-lic
nets has no redundant transitions which can never be executed
from the initial marking.

\begin{definition}[well-formed \textsc{csa}-net]
\label{def:10dggdd}
    A \textsc{csa}-net  is \emph{well-formed}
    if  each transition occurs in at least one
    step sequence and
    all step sequences are well-formed.
    \\
    Notation:
    $\WFCSAN$ is the set of all well-formed \textsc{csa}-nets.
\EOD
\end{definition}
 
Intuitively, the executions of syn-cycles in \textsc{csa}-nets correspond to the 
executions of individual transitions in acyclic nets. 
 
\begin{proposition}
\label{prop:3xxddd}
    All step sequences of a  
    \textsc{bdcsa}-net are well-formed.
\end{proposition}
\begin{proof}
    Similar to the proof of Proposition~\ref{prop:3xx}
    using Proposition~\ref{prop-uhddegr}.  
\qed
\end{proof}

\begin{proposition}
\label{prop:3x}
    $\CSON\subset\BDCSAN\subset\WFCSAN\subset\CSAN$.
\end{proposition}
\begin{proof}
    If follows from Proposition~\ref{prop:3}.
\qed
\end{proof}

\section{Scenarios in \textsc{csa}-nets}

Scenarios of \textsc{csa}-net  start at the same initial 
marking and are both backward
and forward de\-ter\-mi\-nistic. As a result, they represent
potential executions with clearly de\-ter\-mi\-ned causal
relationships. They are also more abstract than (mixed) step sequences,
as one scenario will in general correspond to many step sequences.
 
\begin{definition}[scenario and maximal scenario of \textsc{csa}-net]
\label{def:11dd} 
    A \emph{scenario} of a \textsc{csa}-net $\csan$ is a \textsc{cso}-net
    $\cson$ such that $\cson\sqsubseteq\csan$. Moreover, $\cson$ is 
    \emph{maximal} if there is no scenario $\cson'\neq\cson$ 
    such that
    $\cson\sqsubseteq\cson'$.
\\ 
    Notation:
    $\scenarios(\csan)$ and $\maxscenarios(\csan)$ are respectively
    the sets of all scenarios and maximal scenarios
    of $\csan$.
\EOD 
\end{definition}
Intuitively, scenarios  represent
possible executions (concurrent histories) consistent with the 
dependencies imposed by the flow relation. This, in particular,
means that all the transitions have been executed and places marked
at some point during an execution.
Maximal scenarios are complete in the sense that
they cannot be extended any further.
 
Although scenarios represent behaviours of \textsc{csa}-nets
in a different way than step sequences, markings and related notions,
there is a close relationship between these two approaches.
First, we can see that scenarios do not generate step sequences nor markings
which were not generated by the original a\-cyc\-lic net.
\begin{proposition}
\label{prop:8dd}
    $\bigcup f(\scenarios(\csan))\subseteq f(\csan)$, for every
    well-formed
    \textsc{csa}-net $\csan$ and $f=\mixsseq,\sseq, \reach$. 
\end{proposition}
\begin{proof}
    If follows from Definition~\ref{def:11dd} and Proposition~\ref{pr-uusussddd}. 
\qed
\end{proof}

\begin{definition}[scenario of well-formed step sequence of \textsc{csa}-net]
\label{def:13dd}
    Let $\sigma$ be a well-formed step
    sequence  of a \textsc{csa}-net 
    $\csan= (\an_1,\dots,\an_n,Q_\csan,W_\csan)$.
    Moreover, let $T=\bigcup\sigma$ and
    $\on_i=\scenario_{\an_i}(\proj_{T_{\an_i}}(\sigma))$,
    for every $1\leq i\leq n$.
    Then  
\[
    \scenario_\csan(\sigma)\triangleq
       (\on_1,\dots,\on_n,
        \POST{\csan}{T}\cap Q_\csan,W_\csan|_{(Q\times T)\cup(T\times Q)}) 
\]
    is the \emph{scenario} of $\csan$ generated by $\sigma$.
\EOD
\end{definition}
 
Scenarios of well-formed step sequences
are scenarios in the original structural sense.

\begin{proposition}
\label{pr-icicieddd}
    Let $\sigma$  be a well-formed step sequence
    of a \textsc{csa}-net $\csan$.
\begin{enumerate}
\item
    $\scenario_\csan(\sigma)\in\scenarios(\csan)$
    and
    $\sigma\in\maxsseq(\scenario_\csan(\sigma))$.
\item
    $\sigma\in\maxsseq(\csan)$
    implies $\scenario_\csan(\sigma)\in\maxscenarios(\csan)$.
\end{enumerate}
\end{proposition}
\begin{proof}
    Similar to the proof of Proposition~\ref{pr-icicied}.
\qed
\end{proof}

Two well-formed step sequences
generate the same scenario iff
they execute exactly the same transitions.

\begin{proposition}
\label{pr-icicidd}
    Let $\sigma$ and $\sigma'$ be well-formed step sequences
    of  a \textsc{csa}-net $\csan$.
    Then
    $\scenario_\csan(\sigma)=\scenario_\csan(\sigma')$
    iff
    $\bigcup\sigma=\bigcup\sigma'$.
\end{proposition}
\begin{proof}
    It follows directly from Definition~\ref{def:13}.
\qed
\end{proof}

The well-formedness of \textsc{csa}-nets can also be expressed
in terms of scenarios.

\begin{proposition}
\label{prop:10sddd}
    The following statements are equivalent,
    for every \textsc{csa}-net.
\begin{enumerate}
\item
    The \textsc{csa}-net   is well-formed.

\item
    Each transition occurs in at least one
    scenario, and each step sequence is a
    step sequence of at least one scenario. 
\end{enumerate}
\end{proposition}
\begin{proof}
    Similar to the proof of Proposition~\ref{prop:10sd}.
\qed
\end{proof}

Each (maximal) step sequence of a well-formed
\textsc{csa}-net generates a  (maximal)
scenario
and, conversely, each (maximal) scenario  is generated by a
(maximal)
step sequence.

\begin{proposition}
\label{prop:4tredd}
    Let $\csan$ be a well-formed \textsc{csa}-net.
\begin{enumerate}
\item
    $\scenarios(\csan)=\scenario_\csan(\sseq(\csan))$.
\item
    $\maxscenarios(\csan)=\scenario_\csan(\maxsseq(\csan))$.
\end{enumerate}
\end{proposition}

\begin{proof}
    Similar to the proof of Proposition~\ref{prop:4tre}. 
\qed
\end{proof}

There is a close correspondence between
scenarios and step sequences of a well-formed
\textsc{csa}-net   $\csan$. In particular,
the behaviour of the \textsc{csa}-net corresponds to the joint behaviour of
its scenarios.

\begin{proposition}
\label{prop:4trefffdd}
    Let $\csan$ be a well-formed \textsc{csa}-net.
\begin{enumerate}
\item
    $f(\csan) =\bigcup f(\scenarios(\csan))$, for
    $f=\sseq,\mixsseq,\reach$.
\item
    $f(\csan) =\bigcup f(\maxscenarios(\csan))$, for
    $f=\maxsseq,\maxmixsseq,\finreach$.
\end{enumerate}
\end{proposition}

\begin{proof}
    Similar to the proof of Proposition~\ref{prop:4trefff}. 
\qed
\end{proof}

\begin{example}
\label{ex-3d3}
Consider the \textsc{csa}-net $\csan_1$ in Figure~\ref{fig-2x}($a$).
Then
$\scenario_{\csan_1}(\{a,e\})=\cson_1$,
$\scenario_{\csan_1}(eab)=\cson_2$, and
$\scenario_{\csan_1}(\{c,e\}\{d,f\})=\cson_3$.
\EOD
\end{example}
 
\begin{definition}[syn-cycle of \textsc{csa}-net]
\label{def-syncycledd}
    The \emph{syn-cycles} of a well-formed 
    \textsc{csa}-net $\csan$ are  given by
    $\syncycles(\csan)=\bigcup\{\syncycles(\cson)\mid\cson\in\scenarios(\csan)\}$.
    \\
    Notation:
    $\syncycles(\csan)$ is the set of all sync-cycles of $\csan$.
\EOD
\end{definition}
 
\begin{example}
\label{ex-33dddd}
For $\csan_1$ in Figure~\ref{fig-2x}($a$),
$\syncycles(\cson_3)=\{\{a\},\{b\},\{c\},\{e\},\{d,f\}\}$.
\EOD
\end{example} 

Each step occurring in a step sequence of a \textsc{csa}-net 
can be partitioned into syn-cycles (in a unique way). 

\begin{proposition}
\label{prop-uhddegrdd}
     Let $M$ be a reachable marking of a well-formed \textsc{csa}-net $\csan$
     and $M\STEP{U}{\csan}M'$.
     Then there are syn-cycles
     $S_1,\dots, S_k\in\syncycles(\csan)$ 
     such that $U=S_1\uplus\dots\uplus S_k$ and 
     $M\STEP{S_1\dots S_k}{\csan}M'$.
\end{proposition}
\begin{proof}
     If follows directly from Propositions~\ref{prop-uhddegr} and~\ref{prop:4trefffdd}. 
\qed
\end{proof}

This means, for example, that all reachable markings of a 
well-formed \textsc{csa}-net
can be generated by executing syn-cycles rather than all the potential steps.

The next result can be interpreted as stating that
no part of a well-formed \textsc{csa}-net is semantically irrelevant.

\begin{definition}[coverability and non-redundancy]
\label{def:14} 
    A well-formed \textsc{csa}-net is \emph{covered by scenarios}
    if
    $X_\csan= \bigcup \{X_\cson \mid \cson\in\maxscenarios(\csan)\}$,
    for    
    $X=P,T,F,Q,W$. 
    Moreover,
    $\csan$ has \emph{non-redundant transitions}
    if each transition of $\csan$ occurs in at least one step
    sequence.
\EOD
\end{definition}

\begin{proposition}
\label{prop:14}
    Let $\csan$
    be a well-formed \textsc{csa}-net.
\begin{enumerate}
\item
    $\csan$ is  covered by scenarios
    iff
    it has non-redundant transitions.
\item
    \textsc{cso}-nets have non-redundant transitions.
\end{enumerate}
\end{proposition}
The two ways of capturing the idea that
there are no irrelevant components in $\csan$ are equivalent,
and \textsc{cso}-nets have no redundant components.

The next result shows that an occurrence net has exactly 
one maximal scenario (itself),
which can be interpreted as saying it is a precise representation 
of a single execution history.
\begin{proposition}
\label{prop:9dd}
    $\maxscenarios(\cson)=\{\cson\}$, for every \textsc{cso}-net $\cson$.
\end{proposition}
\begin{proof}
    By Definitions~\ref{def:4aaa}(1,2) and~\ref{def:11dd}(1), we have
    $\cson\in\scenarios(\cson)$. Hence, by Definition~\ref{def:11dd}(2),
    $\cson\in\maxscenarios(\cson)$.

    Suppose now that $\cson'\in\maxscenarios(\on)\setminus\{\cson\}$. Then,
    we have $\cson'\sqsubseteq\cson$ and $\cson'\neq\cson$.
    This yields a contradiction with Definition~\ref{def:11dd}(2).
\qed
\end{proof}

\section{Behavioural structured a\-cyc\-lic nets}

Behavioural Structured Acyclic Nets (\textsc{bsa}-nets)
extend the model of \textsc{bso}-nets introduced and discussed in~\cite{c1,c4,c20,c11,c21,c23}. 
They
can be seen as a way of capturing the evolution of \textsc{csa}-nets, by  
providing a mechanism to abstract parts of a complex activity by another (simpler) system. 
The behaviour in this case is expressed at two levels, namely the upper-level and lower-level. 
The former provides a simple view and hides unimportant details of the behaviour, while 
the latter shows the full details of behaviour during different evolution stages.

\begin{figure}[ht!]
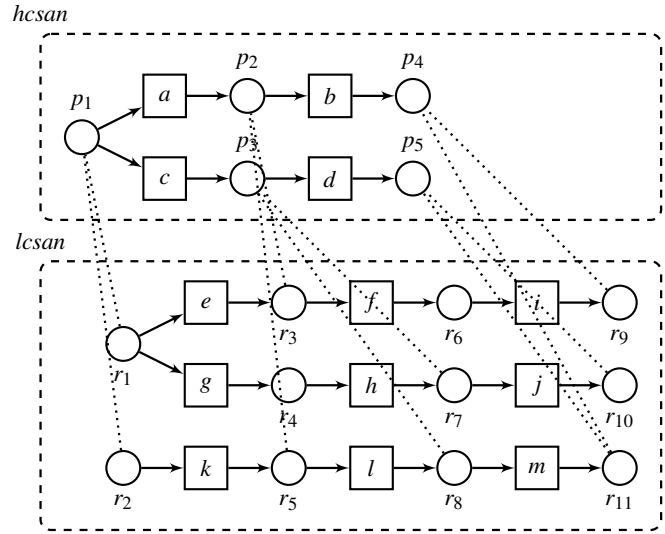

\begin{center}
\StandardNet[0.55]

\Put{-2}{11}{$\hcsan$}
\draw [dashed,rounded corners]
   (-2,10.5) --
   (13,10.5) --
   (13,6) --
   (-2,6) -- cycle;

    \PlacN{P1}{-1}{8}{}{p_1}
    \PlacN{P2}{ 3}{9}{}{p_2}
    \PlacN{P3}{ 3}{7}{}{p_3}
    \PlacN{P4}{ 7}{9}{}{p_4}
    \PlacN{P5}{ 7}{7}{}{p_5}

    \Whitetran{A}{ 1}{9}{a}
    \Whitetran{B}{ 5}{9}{b}
    \Whitetran{C}{ 1}{7}{c}
    \Whitetran{D}{ 5}{7}{d}

    \diredge{P1}{A}
    \diredge{P1}{C}
    \diredge{P2}{B}
    \diredge{P3}{D}

    \diredge{A}{P2}
    \diredge{B}{P4}
    \diredge{D}{P5}
    \diredge{C}{P3}

\Put{-2}{5.5}{$\lcsan$}
\draw [dashed,rounded corners]
   (-2,5) --
   (13,5) --
   (13,-1.5) --
   (-2,-1.5) -- cycle;

    \PlacS{lP1}{ 0}{3}{}{r_1}
    \PlacS{lP2}{ 4}{4}{}{r_3}
    \PlacS{lP3}{ 4}{2}{}{r_4}
    \PlacS{lP4}{ 8}{4}{}{r_6}
    \PlacS{lP5}{ 8}{2}{}{r_7}
    \PlacS{lP6}{12}{4}{}{r_9}
    \PlacS{lP7}{12}{2}{}{r_{10}}
    \PlacS{lQ1}{ 0}{0}{}{r_2}
    \PlacS{lQ2}{ 4}{0}{}{r_5}
    \PlacS{lQ3}{ 8}{0}{}{r_8}
    \PlacS{lQ4}{12}{0}{}{r_{11}}

    \Whitetran{lA}{ 2}{4}{e}
    \Whitetran{lB}{ 6}{4}{f}
    \Whitetran{lC}{ 2}{2}{g}
    \Whitetran{lD}{ 6}{2}{h}
    \Whitetran{lE}{10}{4}{i}
    \Whitetran{lF}{10}{2}{j}
    \Whitetran{lK}{ 2}{0}{k}
    \Whitetran{lL}{ 6}{0}{l}
    \Whitetran{lM}{10}{0}{m}

    \diredge{lP1}{lA}
    \diredge{lP1}{lC}
    \diredge{lP2}{lB}
    \diredge{lP3}{lD}
    \diredge{lP4}{lE}
    \diredge{lP5}{lF}
    \diredge{lQ1}{lK}
    \diredge{lQ2}{lL}
    \diredge{lQ3}{lM}

    \diredge{lA}{lP2}
    \diredge{lB}{lP4}
    \diredge{lD}{lP5}
    \diredge{lC}{lP3}
    \diredge{lE}{lP6}
    \diredge{lF}{lP7}
    \diredge{lK}{lQ2}
    \diredge{lL}{lQ3}
    \diredge{lM}{lQ4}

    \dirDOTS{lP1}{P1}
    \dirDOTS{lQ1}{P1}
    \dirDOTS{lP2}{P2}
    \dirDOTS{lQ2}{P2}
    \dirDOTS{lP5}{P3}
    \dirDOTS{lQ3}{P3}
    \dirDOTS{lP6}{P4}
    \dirDOTS{lQ4}{P4}
    \dirDOTS{lP7}{P5}
    \dirDOTS{lQ4}{P5}

\end{tikzpicture}

\end{center}
\caption{Behavioural structured a\-cyc\-lic net $\bsan_0$ with $n=1$
(i.e., there is one lower-level a\-cyc\-lic net and one upper-level
a\-cyc\-lic net and so graphically there is no difference between $\lcsan$ and
$\lan_1$ as well as between $\hcsan$ and
$\han_1$).}
\label{beh-fig-2x}
\end{figure}

\begin{definition}[\textsc{bsa}-net]
\label{beh-def:20}
    A \emph{be\-ha\-vioural structured a\-cyc\-lic net (or \textsc{bsa}-net)}
    is a triple
    $\bsan\triangleq (\lcsan,\hcsan,\beta)$
    such that:
\[
    \lcsan=(\lan_1,\dots,\lan_n,lQ,lW)
    ~~~\textit{and}~~~
    \hcsan=(\han_1,\dots,\han_n,hQ,hW)
\]
    are well-formed
    \textsc{csa}-nets ($n\geq 1$) with disjoint sets of nodes,
    and
\[
    \beta\subseteq \bigcup_{i=1}^n P_{\lan_i}\times P_{\han_i}
\]
    is a relation  
    satisfying the following, for all $1\leq i\leq n$ and $t\in T_{\han_i}$
    (below $\beta_p=\beta_{\{p\}}\triangleq\{r\mid r\beta p\}$,
    for every $p\in P_\hcsan$):
\begin{enumerate}
\item
    $|M^\init_{\han_i}|=1$ and
    $|\PRE{\han_i}{t}|=|\POST{\han_i}{t}|=1$.

\item
    $\beta_{M^\init_{\han_i}}=M^\init_{\lan_i}$ and
    $\beta_{\PRE {\han_i}{t}}\STEP{}{\lan_i}
     \beta_{\POST{\han_i}{t}}$.
\end{enumerate}
    Notation: $\BSAN$ is the set of all \textsc{bsa}-nets.
\EOD
\end{definition}
Intuitively, $\lcsan$ provides a `lower-level' view
and $\hcsan$ provides a `upper-level' of a record of system behaviour.
The role of $\beta$ is to identify in the lower-level view the
divisions of behaviours into `phases', and each $\beta_p$
indicates a `boundary' between two consecutive phases.

\begin{proposition}
\label{pr-uud}
    Let $\bsan$ be a \textsc{bsa}-net as in Definition~\ref{beh-def:20},
    and let $1\leq i\leq n$.
\begin{enumerate}
\item
    $\reach(\han_i)=\{\{p\}\mid p\in P_{\han_i}\}$.

\item
    $|M\cap P_{\han_i}|=1$, for every $M\in\reach(\hcsan)$.

\item
    $\beta_p\in \reach(\lan_i)$, for every $p\in P_{\han_i}$.
\end{enumerate}
\end{proposition}
That is, a reachable marking of the upper-level \textsc{csa}-net
includes exactly one place from each of its component a\-cyc\-lic nets.
Moreover, all the boundaries between different phases are reachable markings
of the lower-level a\-cyc\-lic nets.

A \textsc{bsa}-net is underpinned by a \textsc{csa}-net
combining the two different views of a database.

\begin{definition}[\textsc{csa}-net underlying \textsc{bsa}-net]
\label{beh-def:20aa}
    Let $\bsan$ be a \textsc{bsa}-net as in Definition~\ref{beh-def:20},
\[
    \Csan(\bsan)
    \triangleq
    (\lan_1,\dots,
    \lan_n,\han_1,\dots,\han_n,lQ\cup hQ,lW\cup hW)\;.
\]
    is the \textsc{csa}-net \emph{underlying} $\bsan$.
\EOD
\end{definition}

\begin{proposition}
\label{pr-uuddd}
    $\Csan(\bsan)$ is a well-formed \textsc{csa}-net.
\end{proposition}
\begin{proof}
    It follows directly from the definitions.
\qed
\end{proof}

\section{Semantics of \textsc{bsa}-nets}
 
The notions concerning the
structure and semantics of \textsc{bsa}-nets introduced below
are based on similar notions defined their underlying 
\textsc{csa}-nets.

\begin{definition}[notation for \textsc{bsa}-net]
\label{defs:d5x}
Let $\bsan$ be a \textsc{bsa}-net as in Definition~\ref{beh-def:20}.
\begin{enumerate}
\item
    $X_\bsan\triangleq X_{\Csan(\bsan)}$, for $X=P,T,F,Q,W$.
\item
    $\PRE{\bsan}{}\triangleq\PRE{\Csan(\bsan)}{}$ and
    $\POST{\bsan}{}\triangleq\POST{\Csan(\bsan)}{}$.
\item
    $\STEPS(\bsan)\triangleq \STEPS(\Csan(\bsan))$.
\item
    $\MARK(\bsan)\triangleq \MARK(\Csan(\bsan))$.
\item
    $M_\bsan^\init\triangleq M_{\Csan(\bsan)}^\init$
    and
    $M_\bsan^\fin\triangleq M_{\Csan(\bsan)}^\fin$.
\EOD
\end{enumerate}
\end{definition}
The notion of enabled step is different.
It is based on the phases of lower-level a\-cyc\-lic nets
induced by the phase boundaries
iduced by the relation~$\beta$.

\begin{definition}[phase and phase-consistent marking of \textsc{bsa}-net]
\label{beh-def:20da}
    Let $\bsan$ be a \textsc{bsa}-net as in Definition~\ref{beh-def:20}.
\begin{enumerate}
\item
    The \emph{phase} of a place $p\in P_{\han_i}$  ($1\leq i\leq n$) is given by:
\[ 
    \phase(p)
    \triangleq 
    \{\beta_p\}\cup
    \{M\mid \exists t\in\POST{\han_i}{p}: 
    \beta_p\STEP{}{\lan_i}M\STEP{}{\lan_i}\beta_{\POST{\han_i}{t}}\}\;. 
\] 

\item
    A marking $M\in\reach(\Csan(\bsan))$ is \emph{phase-consistent}
    if, for every $1\leq i\leq n$:
\[
    M\cap P_{\lan_i}\in\phase(\beta_{M\cap P_{\han_i}})\;.
\]
\end{enumerate}
Notation: $\phcMARK(\bsan)$ is the set of all phase-consistent markings.
\EOD
\end{definition}
Intuitively,  $\phase(p)$ is a contiguous `chunk' of $\lan_i$  delimited
by the marking corresponding to $p$ (start) and all markings (ends)
corresponding to the places obtained by executing one output transition
of $p$ (such a transition indicates a `phase change'). All markings
in-between belong to the delimited phase.
Moreover, phase-consistency means that the markings of the lower-level
a\-cyc\-lic nets belong to the phases corresponding to the markings of the
upper-level a\-cyc\-lic nets.
 
\begin{example}
\label{ex-44r}
    Figure~\ref{beh-fig-2x} depicts a \textsc{bsa}-net $\bsan_0$.
    We have:
\[
\begin{array}{lcl}
    \phase(p_1)
    & =
    & \{ \{r_1,r_2\}, \{r_1,r_5\}, \{r_3,r_2\}, \{r_3,r_5\}, 
         \{r_1,r_8\}, 
         \{r_4,r_2\}, \{r_4,r_5\}, 
         \\
         &&
         ~~\{r_4,r_8\},
         \{r_7,r_2\}, \{r_7,r_5\}, \{r_7,r_8\} \}
\\
    \phase(p_2)
    & =
    &  \{\{r_3,r_5\}, \{r_3,r_8\}, \{r_3,r_{11}\},
    \{r_6,r_5\}, \{r_6,r_8\}, \{r_6,r_{11}\},
    \\
    &&
    ~~\{r_9,r_5\}, \{r_9,r_8\}, \{r_9,r_{11}\}\}
\\
    \phase(p_3)
    & =
    &  \{ \{r_7,r_8\}, \{r_7,r_{11}\},\{r_{10},r_8\}, \{r_{10},r_{11}\} \}
\\
    \phase(p_4)
    & =
    &  \{\{r_9,r_{11}\}\}
\\
    \phase(p_4)
    & =
    &  \{\{r_{10},r_{11}\}\}.
\end{array}
\]
The   marking $\{r_1,r_{11}\}\in\reach(\lan_1)$ does not belong to any phase
of the places of $\han_1$.
The marking
$\{p_1, r_9,r_{11}\}\in\reach(\Csan(\bsan))$ is not phase-consistent, and
the marking
$\{p_4, r_9,r_{11}\}\in\reach(\Csan(\bsan))$ is phase-consistent.
\EOD
\end{example}

The execution semantics of \textsc{bsa}-nets is restricted
to the phase-consistent markings.

\begin{definition}[enabled and executed step of \textsc{bsa}-net]
\label{def:6wwx}
    Let $M,M'\in\MARK(\bsan)$ be markings of a \textsc{bsa}-net $\bsan$,
    and $U\in\STEPS(\bsan)$ be a step.
    Then
\[
    U\in\enabled_\bsan(M)
    ~~~\textit{and}~~~M\STEP{U}{\bsan}M'
\]
    if $M,M'\in\phcMARK(\bsan)$ and
    $M\STEP{U}{\Csan(\bsan)}M'$.
\EOD
\end{definition}
Intuitively, $\bsan$ is executed in exactly the same way as its
underlying \textsc{csa}-net provided that the
the markings involved
are phase-consistent.

\begin{definition}[mixed step sequence and step sequence of \textsc{bsa}-net]
\label{def:7aff}
    Let $\bsan$ be a \textsc{bsa}-net and
    $\mu=M_0U_1M_1\dots M_{k-1}U_k M_k$ ($k\geq 0$)
    be a sequence such that $M_0,\dots,M_k\in\MARK(\bsan)$ and
    $U_1,\dots,U_k\in\in\STEPS(\bsan)$.
\begin{enumerate}
\item
    $\mu$ is a \emph{mixed step sequence from $M_0$ to $M_k$}
    if $M_{i-1}\STEP{U_i}{\bsan}M_i$,
    for every $1\leq i\leq k$.

\item
    If $\mu$ is a mixed step sequence from $M_0$ to $M_k$,
    then $\sigma=U_1\dots U_k$ is
    a \emph{step sequence from $M_0$ to $M_k$}.
\end{enumerate}
    This is denoted by $M_0\mixSTEP{\mu}{\bsan}M_k$
    and $M_0\STEP{\sigma}{\bsan}M_k$, respectively.
    Also, $M_0\STEP{}{\bsan}M_k$
    denotes that $M_k$ is reachable from $M_0$.
\EOD
\end{definition}

In the last definition, the starting point
is an arbitrary marking.
The next definition assumes that the starting point is
the default initial marking.

\begin{definition}[behaviour of \textsc{bsa}-net]
\label{def:fff7}
    The following sets capture various be\-ha\-vioural notions
    related to step sequences and reachable markings
    of a \textsc{bsa}-net $\bsan$.
\begin{enumerate}
\item
    $\sseq(\bsan)\triangleq\{\sigma\mid  M_\bsan^\init\STEP{\sigma}{\bsan}M\}$
    \hfill \emph{step sequences}.

\item
    $\mixsseq(\bsan)\triangleq\{\mu\mid  M_\bsan^\init\mixSTEP{\mu}{\bsan}M\}$
    \hfill   \emph{mixed step sequences}.

\item
    $\maxsseq(\bsan)\triangleq\{\sigma\in\sseq(\bsan)\mid
         \neg\exists U: \sigma U\in\sseq(\bsan)\}$
    \\\hspace*{\fill}   \emph{maximal step sequences}.

\item
    $\maxmixsseq(\bsan)\triangleq\{\mu\in\mixsseq(\bsan)\mid
         \neg\exists U,M: \mu UM\in\mixsseq(\bsan)\}$
    \\\hspace*{\fill} \emph{maximal mixed step sequences}.

\item
    $\reach(\bsan)\triangleq \{M\mid
               M_\bsan^\init\STEP{}{\bsan}M\}$
    \hfill   \emph{reachable markings}.

\item
    $\finreach(\bsan)\triangleq\{M\mid \exists \sigma\in\maxsseq(\bsan):
               M_\bsan^\init\STEP{\sigma}{\bsan}M\}$
    \\\hspace*{\fill}  \emph{final reachable markings}.
\end{enumerate}
\EOD
\end{definition}

\begin{example}
\label{ex-oodo}
A maximal mixed step sequence of
the \textsc{bsa}-net depicted in Figure~\ref{beh-fig-2x} is:
\[
\begin{array}{lclllllllll}
    \mu
    & =
    &
    \{p_1,r_1,r_2\}&\{g,k\}&
    \{p_1,r_4,r_5\}&\{h,l\}&
    \{p_1,r_7,r_8\}&\{c,m\}
\\
   &&
    \{p_3,r_7,r_{11}\}&\{j\}&
    \{p_3,r_{10},r_{11}\}&\{d\}&
    \{p_5,r_{10},r_{11}\}\;,
\end{array}
\]
and the corresponding maximal
step sequence is $\sigma=\{g,k\} \{h,l\} \{c,m\} \{j\} \{d\}$.
\EOD
\end{example}

As before, we single out nets  with forward and
backward de\-ter\-mi\-nism.

\begin{definition}[\textsc{bso}-net]
\label{ddef:22}
    A \textsc{bsa}-net $\bson=(\lcson,\hcson,\beta)$  
    is a
    \emph{be\-ha\-vioural structured occurrence net (or \textsc{bso}-net)}
    if $\lcson,\hcson\in\CSON$ and there is 
    $\sigma\in\sseq(\bson)$ such that $\bigcup\sigma=T_\bson$.
    \\
    Notation:
    $\BSON$ is the set of all \textsc{bso}-nets.
\EOD
\end{definition}
Note that in \textsc{bso}-nets, the upper-level a\-cyc\-lic nets
are `line-like' occurrence nets.

As before, what really matters is to identify in a \textsc{bsa}-net all the de\-ter\-mi\-nistic
behaviours (scenarios).

\begin{definition}[scenario and maximal scenario of \textsc{bsa}-net]
\label{def:11dx}
    Let $\bsan$ be a \textsc{bsa}-net as in Definition~\ref{beh-def:20}.
\begin{enumerate}
\item
    A \emph{scenario} of  $\bsan$ is a \textsc{bso}-net
    $\bson=(\lcson,\hcson,\beta')$ such that:
\begin{enumerate}
\item
    $\lcson\in\scenarios(\lcsan)$ and $\hcson\in\scenarios(\hcsan)$.
\item
    $\beta'=\beta\cap(P_\lcson\times P_\hcson)$.
\end{enumerate}
\item
    A \emph{maximal scenario} of $\bsan$ is a scenario $\bson$
    such that there is no scenario $\bson'$ satisfying
    $T_\bson\subset T_{\bson'}$.
\end{enumerate}
    Notation:
    $\scenarios(\bsan)$ and $\maxscenarios(\bsan)$ are respectively
    the sets of all scenarios and maximal scenarios of $\bsan$.
\EOD
\end{definition}

\begin{example}
    Figure~\ref{beh-fig-fddfff2x}
    depicts the only two maximal scenarios of the \textsc{bsa}-net in
    Figure~\ref{beh-fig-2x}.
\end{example}

\begin{figure}[ht!]
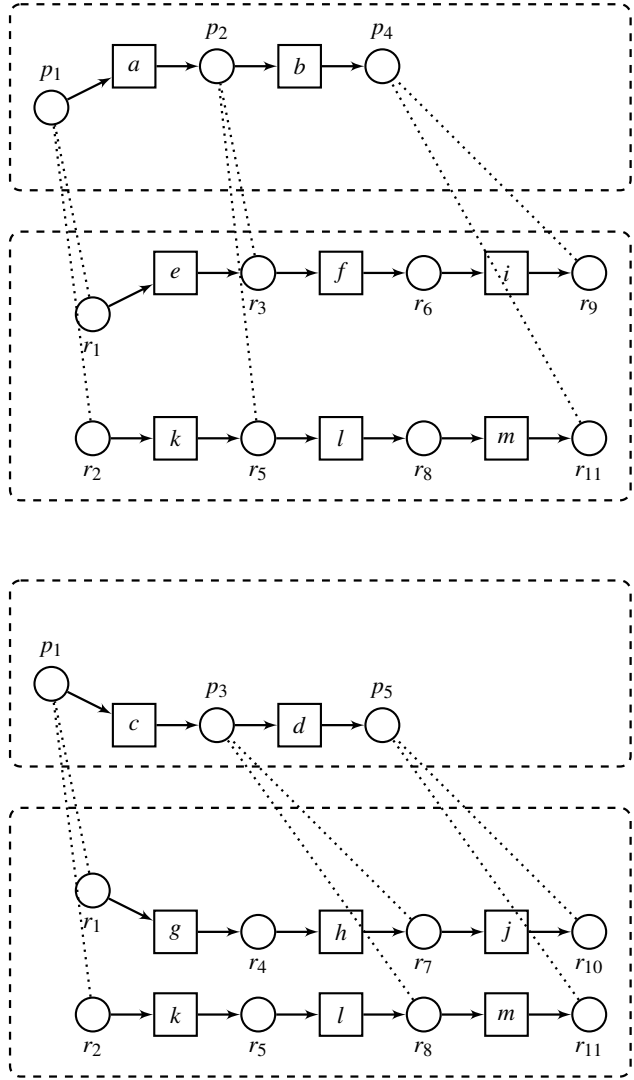

\begin{center}
\StandardNet[0.55]

\draw [dashed,rounded corners]
   (-2,10.5) --
   (13,10.5) --
   (13,6) --
   (-2,6) -- cycle;

    \PlacN{P1}{-1}{8}{}{p_1}
    \PlacN{P2}{ 3}{9}{}{p_2}
    \PlacN{P4}{ 7}{9}{}{p_4}

    \Whitetran{A}{ 1}{9}{a}
    \Whitetran{B}{ 5}{9}{b}

    \diredge{P1}{A}
    \diredge{P2}{B}

    \diredge{A}{P2}
    \diredge{B}{P4}

\draw [dashed,rounded corners]
   (-2,5) --
   (13,5) --
   (13,-1.5) --
   (-2,-1.5) -- cycle;

    \PlacS{lP1}{ 0}{3}{}{r_1}
    \PlacS{lP2}{ 4}{4}{}{r_3}
    \PlacS{lP4}{ 8}{4}{}{r_6}
    \PlacS{lP6}{12}{4}{}{r_9}
    \PlacS{lQ1}{ 0}{0}{}{r_2}
    \PlacS{lQ2}{ 4}{0}{}{r_5}
    \PlacS{lQ3}{ 8}{0}{}{r_8}
    \PlacS{lQ4}{12}{0}{}{r_{11}}

    \Whitetran{lA}{ 2}{4}{e}
    \Whitetran{lB}{ 6}{4}{f}
    \Whitetran{lE}{10}{4}{i}
    \Whitetran{lK}{ 2}{0}{k}
    \Whitetran{lL}{ 6}{0}{l}
    \Whitetran{lM}{10}{0}{m}

    \diredge{lP1}{lA}
    \diredge{lP2}{lB}
    \diredge{lP4}{lE}
    \diredge{lQ1}{lK}
    \diredge{lQ2}{lL}
    \diredge{lQ3}{lM}

    \diredge{lA}{lP2}
    \diredge{lB}{lP4}
    \diredge{lE}{lP6}
    \diredge{lK}{lQ2}
    \diredge{lL}{lQ3}
    \diredge{lM}{lQ4}

    \dirDOTS{lP1}{P1}
    \dirDOTS{lQ1}{P1}
    \dirDOTS{lP2}{P2}
    \dirDOTS{lQ2}{P2}
    \dirDOTS{lP6}{P4}
    \dirDOTS{lQ4}{P4}

\end{tikzpicture}
\\[1cm]

\StandardNet[0.55]

\draw [dashed,rounded corners]
   (-2,10.5) --
   (13,10.5) --
   (13,6) --
   (-2,6) -- cycle;

    \PlacN{P1}{-1}{8}{}{p_1}
    \PlacN{P3}{ 3}{7}{}{p_3}
    \PlacN{P5}{ 7}{7}{}{p_5}
    \Whitetran{C}{ 1}{7}{c}
    \Whitetran{D}{ 5}{7}{d}

    \diredge{P1}{C}
    \diredge{P3}{D}

    \diredge{D}{P5}
    \diredge{C}{P3}

\draw [dashed,rounded corners]
   (-2,5) --
   (13,5) --
   (13,-1.5) --
   (-2,-1.5) -- cycle;

    \PlacS{lP1}{ 0}{3}{}{r_1}
    \PlacS{lP3}{ 4}{2}{}{r_4}
    \PlacS{lP5}{ 8}{2}{}{r_7}
    \PlacS{lP7}{12}{2}{}{r_{10}}
    \PlacS{lQ1}{ 0}{0}{}{r_2}
    \PlacS{lQ2}{ 4}{0}{}{r_5}
    \PlacS{lQ3}{ 8}{0}{}{r_8}
    \PlacS{lQ4}{12}{0}{}{r_{11}}

    \Whitetran{lC}{ 2}{2}{g}
    \Whitetran{lD}{ 6}{2}{h}
    \Whitetran{lF}{10}{2}{j}
    \Whitetran{lK}{ 2}{0}{k}
    \Whitetran{lL}{ 6}{0}{l}
    \Whitetran{lM}{10}{0}{m}

    \diredge{lP1}{lC}
    \diredge{lP3}{lD}
    \diredge{lP5}{lF}
    \diredge{lQ1}{lK}
    \diredge{lQ2}{lL}
    \diredge{lQ3}{lM}

    \diredge{lD}{lP5}
    \diredge{lC}{lP3}
    \diredge{lF}{lP7}
    \diredge{lK}{lQ2}
    \diredge{lL}{lQ3}
    \diredge{lM}{lQ4}

    \dirDOTS{lP1}{P1}
    \dirDOTS{lQ1}{P1}
    \dirDOTS{lP5}{P3}
    \dirDOTS{lQ3}{P3}
    \dirDOTS{lP7}{P5}
    \dirDOTS{lQ4}{P5}

\end{tikzpicture}

\end{center}
\caption{Maximal scenarios for the \textsc{bsa}-net in Figure~\ref{beh-fig-2x}.}
\label{beh-fig-fddfff2x}
\end{figure}

\section{Well-formed  \textsc{bsa}-nets}

The general definition of \textsc{bsa}-net given above does not
guarantee that the step sequences of $\bsan$ cover all possible scenarios
of
the lower-level \textsc{csa}-net. Indeed, in the extreme case,
we can take $\hcsan$ which contains no transitions
at all, and the resulting \textsc{bsa}-net generates then only the
empty step sequence.
It is therefore crucial to identify cases where
$\bsan$  generates at least one step sequence
for every scenario of $\lcsan$.
As a result, the definition of a well-formed
\textsc{bsa}-net is more demanding than the definition of
a well-formed
\textsc{csa}-net.

\begin{definition}[well-formed \textsc{bsa}-net]
\label{def:10wdeeeedx}
    A \textsc{bsa}-net $\bsan$
    is \emph{well-formed}
    if:
\begin{enumerate}
\item
    $\sseq(\bsan)=\bigcup\sseq(\scenarios(\bsan))$.
\item
    For every scenario $\bson\in\maxscenarios(\bsan)$, there is
    $\sigma\in\maxsseq(\bsan)$ such that
    $\sigma\upharpoonright_{T_\bson}\in\maxsseq(\bson)$.
\end{enumerate}
    Notation: $\WFBSAN$ is the set of all well-formed \textsc{bsa}-nets.
\EOD
\end{definition}

\section{Concluding remarks}

This paper provides formalisation and basic properties for the nine classes of nets
listed in the diagram below.

\begin{center}
\begin{tabular}{r@{~}|@{~}l@{~}|@{~}l@{~}|@{~}l@{~}|}
\emph{backward \& forward det}
& occurrence net
& \textsc{cso}-net
& \textsc{bso}-net
\\
\emph{backward  det}
& backward det a\-cyc\-lic net
& \textsc{bdcsa}-net
& \textsc{bdbsa}-net
\\
\emph{no restriction}
& a\-cyc\-lic net
& \textsc{csa}-net
& \textsc{bsa}-net
\end{tabular}
\end{center}

The following are some 
published works on the topics concerned with or related to the material presented above
\cite{c1,c4,c10,c20,c23,c34,c35,c37,c51,c24,c127,c21,c11,c125,c126,c54,DBLP:conf/apn/Alharbi23,DBLP:conf/apn/Alshammari23,DBLP:conf/apn/Almutairi23,DBLP:conf/apn/Alahmadi23,DBLP:journals/topnoc/BhattacharyyaK23}.

\end{document}